\documentclass[11pt]{article}

\usepackage[margin=1in]{geometry}
\usepackage[T1]{fontenc}
\usepackage[utf8]{inputenc}
\usepackage{lmodern}
\usepackage{microtype}
\usepackage{amsmath,amssymb,amsthm,mathtools}
\usepackage{enumitem}
\usepackage{hyperref}
\usepackage{framed}
\usepackage{tikz}
\usetikzlibrary{arrows.meta,calc,positioning,fit,decorations.pathreplacing}

\hypersetup{
  colorlinks=true,
  linkcolor=blue,
  citecolor=blue,
  urlcolor=blue,
  pdftitle={Uniformity Testing for the Monomer-Dimer Model}
}

\newtheorem{theorem}{Theorem}[section]
\newtheorem{proposition}[theorem]{Proposition}
\newtheorem{lemma}[theorem]{Lemma}
\newtheorem{remark}[theorem]{Remark}
\newtheorem{corollary}[theorem]{Corollary}

\newcommand{\Mcal}{\mathcal{M}}
\newcommand{\E}{\mathbb{E}}
\newcommand{\Pbb}{\mathbb{P}}
\newcommand{\Var}{\mathrm{Var}}
\newcommand{\Cov}{\mathrm{Cov}}
\newcommand{\Ent}{\mathrm{Ent}}
\newcommand{\TV}{d_{\mathrm{TV}}}
\newcommand{\Ber}{\mathrm{Bernoulli}}
\newcommand{\one}{\mathbf{1}}
\newcommand{\R}{\mathbb{R}}
\newcommand{\KL}{D_{\mathrm{KL}}}
\newcommand{\Lcal}{\mathcal{L}}
\newcommand{\ip}[2]{\langle #1,#2\rangle}
\newcommand{\norm}[1]{\lVert #1\rVert}
\newcommand{\F}{\mathrm{F}}

\title{Uniformity Testing for the Monomer-Dimer Model}
\author{
Joy Qiping Yang\thanks{This project was carried out during his participation in the 2026 NII International Internship Program.}\\
University of Sydney\\
\texttt{qyan6238@uni.sydney.edu.au}
\and  
Yuichi Yoshida\thanks{Supported by JSPS KAKENHI Grant Number JP24K02903.}\\
National Institute of Informatics\\
\texttt{yyoshida@nii.ac.jp}
}
\date{}

\begin{document}
\maketitle

\begin{abstract}
We study uniformity testing for the monomer-dimer model on a known graph
$G=(V,E)$.  Given samples from an unknown distribution
$D_w(M)\propto \prod_{e\in M} w_e$ over matchings $M$ of $G$, we ask whether
$D_w$ is the uniform distribution over all matchings of $G$ or is
$\varepsilon$-far from it in total variation distance.

Our main result is that, for every maximum-degree bound $\Delta$, uniformity can
be tested with $\widetilde O_\Delta(\sqrt m/\varepsilon^2)$ samples, where
$m=|E|$.  We complement this with an
information-theoretic lower bound for general graphs: for every pair of
parameters $n\le m\le n^2$, there are graphs with
$\Theta(n)$ vertices and $\Theta(m)$ edges on which every tester requires
$\Omega(\sqrt m/\varepsilon^2)$ samples.
Thus, for bounded-degree graph families with $m=\Theta(n)$, testing uniformity over the
space of matchings has
$\widetilde \Theta(\sqrt n/\varepsilon^2)$ sample complexity, in contrast with the
$\Theta(n/\varepsilon^2)$ samples needed to estimate $D_w$ in total variation
distance.
\end{abstract}

\thispagestyle{empty}
\newpage

\thispagestyle{empty}
\setcounter{tocdepth}{2}
\tableofcontents
\newpage

\setcounter{page}{1}

\section{Introduction}

Distribution testing asks how many samples are needed to decide whether an unknown distribution
has a prescribed property.  Many basic questions are now well understood over unstructured finite
domains; see, for example, the surveys~\cite{Rubinfeld12, Canonne20, Canonne22}
and the early work of Batu, Fortnow, Rubinfeld, Smith, and White~\cite{BatuFRSW00}.  In modern
applications, however, the objects being sampled are often high-dimensional.  Without structural
assumptions, even basic testing and learning tasks over domains such as $\{0,1\}^n$ require sample
complexity exponential in the dimension (see, e.g., \cite[Appendix H.2]{AcharyaDK15} and
\cite[Theorem B.1]{BhattacharyyaCY22}).

This curse of dimensionality has motivated a broad line of work on \emph{structured}
high-dimensional distribution testing.  For families such as Bayesian networks, Ising models, and
other graphical models, low-dimensional local or conditional structure can be used to detect
global distance from the target distribution, and more broadly has led to sample- and
time-efficient algorithms for testing and learning in important
regimes~\cite{CDKS20, DDK18, BhattacharyyaGGSV25,BhattacharyyaGPTV23}.
A natural next question is whether these techniques extend beyond product-like domains:
\emph{can one test natural high-dimensional distributions whose support itself is defined by hard
combinatorial constraints?}

We study this question for the \emph{monomer-dimer model}, equivalently the \emph{Gibbs family over all
matchings} of a known graph.  Let $G=(V,E)$ be an undirected graph with $m=|E|$, and let
$w=(w_e)_{e\in E}$ be positive edge weights.  The corresponding Gibbs distribution over matchings
$D_w$ is supported on the set $\Mcal(G)$ of all matchings of $G$ and is defined by
\[
D_w(M)\propto \prod_{e\in M} w_e,
\qquad M\in \Mcal(G).
\]
This is the standard weighted form of the monomer-dimer model from statistical
physics~\cite{HL72}; equivalently, it is the hard-core model on the line graph $L(G)$, whose
vertices are the edges of $G$.  The uniform distribution over all matchings corresponds to
the all-one weight vector $w_e\equiv 1$.  Thus the problem is \emph{uniformity testing} within a
classical $m$-parameter Gibbs family.

The monomer-dimer model also has a direct modeling interpretation.  Vertices represent sites or
agents, edges represent possible pairings or interactions, and a matching records a set of pairings
that can coexist because no vertex can participate in two pairs.  Edge weights then represent local
preferences or interaction intensities.  This interpretation is common in applications to large
communication and social networks, where one is interested in typical matchings and local
participation probabilities~\cite{LelargeZhou14}.

In this interpretation, the uniform distribution over all matchings is the reference distribution
with no edge-specific preference beyond the hard exclusion constraint itself.  Testing uniformity
asks whether observed matchings can be explained by the hard constraint alone, or whether they
reveal nonuniform edge weights.

A sample from $D_w$ is an edge-incidence vector in $\{0,1\}^E$, but not an arbitrary one: it
must lie in the constrained family $\Mcal(G)$.  The support can be exponentially large in $m$: for
example, if $G$ is the disjoint union of $m$ edges, then every subset of edges is a matching and
$|\Mcal(G)|=2^m$.  The coordinates are also strongly dependent
because choosing one edge deterministically rules out all adjacent edges.

The problem of testing uniformity asks whether an unknown member of this family is the uniform
distribution $U$ on $\Mcal(G)$:
\begin{framed}
    \noindent Given sample access to $D_w$ (with $w$ unknown) over a known graph $G=(V,E)$: how many
    samples are needed to distinguish between the case $D_w=U$ and the case
    $\TV(D_w,U)\ge \varepsilon$, with probability at least $2/3$?
\end{framed}
Although this reference distribution is conceptually simple, the tester is not given oracle access to
$U$ or an explicit probability table for it.  The distribution is specified only through the graph
$G$ and the feasible set $\Mcal(G)$, which may have exponential size; sampling from $U$ is a
separate combinatorial problem.

The statistical difficulty is not merely that the sample space can be exponentially large:
ordinary edge marginals are not the right local quantities.  An edge may be absent because its own
weight is small, or simply because a neighboring edge was selected.  Thus the question is whether
the matching constraint still leaves enough conditional structure to test the all-one-weight case
with fewer samples than learning the full edge-weight model.

\subsection{Our results}

Our main result is a polynomial-time uniformity tester for bounded-degree graphs with sample
complexity $\widetilde O_\Delta(\sqrt m/\varepsilon^2)$.  We also prove a matching
information-theoretic lower bound of $\Omega(\sqrt m/\varepsilon^2)$, in fact for graph families of
every density (from bounded to unbounded degrees).

The algorithmic side is the following bounded-degree tester.

\begin{theorem}[Efficient bounded-degree tester]
\label{thm:main}
There is a universal constant $C<\infty$ such that the following holds.  For every integer
$\Delta\ge 1$, there is an explicit constant
\[
F_\Delta\le (C\Delta)^{C\Delta^5}
\]
such that, for every graph $G=(V,E)$ with $\Delta(G)\le \Delta$ and $m=|E|\ge 1$, every
$0<\varepsilon\le 1$, and every $0<\delta<1$, there is a tester which distinguishes
\[
H_0:\ D_w=U
\qquad\text{from}\qquad
H_1:\ \TV(D_w,U)\ge \varepsilon
\]
with error probability at most $\delta$, using at most
\[
F_\Delta\left(
\log\frac{4m}{\delta}
+
\frac{\sqrt m\,\log(2/\delta)}{\varepsilon^2}
\right)
\]
samples from the unknown distribution.  For every fixed $\Delta$, the tester can be implemented in
polynomial time in $m$, $1/\varepsilon$, and $\log(2/\delta)$.
\end{theorem}

For bounded-degree graph families, where $m=\Theta(n)$, Theorem~\ref{thm:main} gives a
$\widetilde O(\sqrt n/\varepsilon^2)$-sample tester.  We do not optimize the explicit dependence on
$\Delta$.

We next state a complementary lower bound, in a form that covers all graph densities between
sparse and dense.

\begin{theorem}[Lower bound at every density]
\label{thm:main-lower}
There are absolute constants $c_0>0$ and $\varepsilon_0>0$ such that the following holds.  For every
pair of integers $n,m$ with $n\le m\le n^2$, there is a graph $G_{n,m}$ satisfying
\[
|V(G_{n,m})|=\Theta(n),
\qquad
|E(G_{n,m})|=\Theta(m),
\]
such that every tester, given i.i.d.\ samples from an unknown Gibbs distribution over matchings of
$G_{n,m}$, distinguishes
\[
H_0:\ D_w=U_{G_{n,m}}
\qquad\text{from}\qquad
H_1:\ \TV(D_w,U_{G_{n,m}})\ge \varepsilon
\]
with success probability at least $2/3$ must use at least
\[
c_0\frac{\sqrt m}{\varepsilon^2}
\]
samples for every $0<\varepsilon\le\varepsilon_0$; and $m = m(n)$, where the expression depends on the degree.
\end{theorem}

Together, Theorems~\ref{thm:main} and~\ref{thm:main-lower} show that $\sqrt m$ is the correct
dependence for bounded-degree uniformity testing in this family.  This is not the classical $\sqrt N$ rate for an
explicit domain of size $N$: the domain is $\Mcal(G)$, whose size can be exponential in $m$.
Rather, testing is governed by the number of edge parameters.  Compared with the
$\Theta(m/\varepsilon^2)$ samples needed to learn the full Gibbs family, this gives a quadratic
testing-versus-learning separation on bounded-degree graphs.

\subsection{Technical overview}

The central issue is that a sample is not an unconstrained vector in $\{0,1\}^E$: once an edge is
present, all incident edges are forced to be absent.  Thus edge marginals alone are a poor object to
compare.  The right local question is instead: if no edge adjacent to $e$ is selected, is $e$ itself
chosen with probability $1/2$, as it would be under $U$?  This turns the matching constraint into local
conditional tests.

\paragraph{Upper bound: from global distance to edgewise conditional statistics.}
The Gibbs family over matchings has a useful rigidity property: $D_w=U$ if and only if $w_e=1$ for
every edge $e$.  This is why it is informative to look at a single edge after the states of all
other edges have been fixed.  For a matching $M$, write
$X_e(M)=\one\{e\in M\}$ and
$A_e(M)=\one\{\text{no edge adjacent to }e\text{ is used in }M\}$.  When $M$ is sampled from
$D_w$ or $U$, we suppress the argument and regard $X_e$ and $A_e$ as random variables.
We also write $\xi_e=A_e(2X_e-1)$.
Conditional on all other coordinates, there are only two cases.  If $A_e=0$, then $X_e=0$ is
forced; if $A_e=1$, then
\[
  D_w(X_e=1\mid X_{-e})=\frac{w_e}{1+w_e},
  \qquad
  U(X_e=1\mid X_{-e})=\frac12,
\]
where $X_{-e} \in \{0, 1\}^{E \setminus \{e\}}$ denotes any fixed configuration other than $e$ (that is feasible for $e$ to be matched, since $A_e = 1$ here).
Writing $t_e=(w_e-1)/(w_e+1)$ and $a_e=D_w(A_e)$, $\xi_e$ has expectation $a_e t_e$ under
$D_w$ and has expectation zero under $U$.
Thus the problem becomes to detect whether the expectations of these statistics are collectively
nonzero.

The availability probabilities $a_e$ are themselves distribution dependent: changing weights near
$e$ can change how often $e$ is available, even if the conditional bias $t_e$ is small.  The
\emph{screening step} estimates $a_e$ and $a_et_e$ from samples.  To calibrate this step, we
compute under $U$
\[
  \alpha_e:=U(A_e)=\frac{|\Mcal(G^{(e)})|}{|\Mcal(G)|}
\]
up to constant relative accuracy using an FPTAS for counting matchings in bounded-degree graphs,
where $G^{(e)}$ is obtained by deleting the edges conflicting with $e$~\cite{BGKNT07}.\footnote{Note that one can also employ the Jerrum-Sinclair algorithm \cite{jerrum1996markov} to estimate the number of matchings, which can bypass the degree restriction here. But the main technical hurdle, which is Proposition \ref{prop:imported-tools}, on the bounded-degree restriction for testing the monomer-dimer model remains.}  Thus the
tester only needs the probabilities $\alpha_e$ under $U$, obtained from matching counts, not an
explicit table or sampler for $U$.  Then a small initial block of samples checks whether the
empirical availability $\widehat a_e$ is close to this value under $U$ and whether the empirical
average
$\widehat{\xi}_e$ of $\xi_e$ is close to zero for every edge.
If either check fails, we have
already detected evidence against uniformity and reject.  If the checks pass, then the analysis proves
more than empirical closeness: it forces $a_e\asymp_\Delta \alpha_e$, $a_e t_e$ small, and in
particular $w_e\in[1/2,2]$ for every edge.

After screening, for each edge we use the normalized statistic
\[
  \psi_e(M):=\frac{\xi_e(M)}{\sqrt{\widetilde\alpha_e}},
  \qquad
  \Psi(M)=(\psi_e(M))_{e\in E},
\]
where $\widetilde\alpha_e$ is the computed approximation to $\alpha_e$.  The normalization makes
the coordinates comparable across edges.  The numerator $\xi_e$ is conditionally unbiased
under $U$, so $\E_U[\Psi]=0$ regardless of the normalizing constants.

The next step is to show that $\E_{D_w}[\Psi]$ is large whenever $D_w$ is far from $U$.  This
follows the graphical-model testing strategy of reducing global distance to the mean of statistics
defined from low-dimensional conditional distributions~\cite{CDKS20,DDK18}.
Approximate tensorization of entropy on the line graph lets us upper bound $\KL(D_w\|U)$ by a
sum of one-edge conditional KL terms.  For edge $e$, this conditional term is zero when
$A_e=0$, because then both $D_w$ and $U$ force $X_e=0$; when $A_e=1$, it is the KL between
$\Ber((1+t_e)/2)$ and $\Ber(1/2)$.  Using spectral-independence and entropy-factorization
results for hard-core measures on claw-free graphs~\cite{CG23,CLV23}, this gives
\[
  \KL(D_w\|U)
  \le C_\Delta^{\mathrm{ent}}
       \sum_{e\in E} a_e\,
       \KL\!\left(\Ber\!\left(\frac{1+t_e}{2}\right)\middle\|\Ber\!\left(\frac12\right)\right).
\]
After the screening step accepts, the Bernoulli KL is quadratic in $t_e$, and the comparison
$a_e\asymp_\Delta\widetilde\alpha_e$ turns this local KL bound into a lower bound on
$\|\E_{D_w}[\Psi]\|_2^2$.  Therefore total variation distance at least $\varepsilon$ forces
\[
  \nu:=\E_{D_w}[\Psi],
  \qquad
  \|\nu\|_2^2\gtrsim_\Delta \varepsilon^2.
\]
Thus we have reduced the problem of testing uniformity to detecting an $m$-dimensional mean vector
whose squared norm is at least a constant multiple of $\varepsilon^2$.

The remaining question is why this mean can be tested with only $\sqrt m/\varepsilon^2$ samples
rather than $m/\varepsilon^2$.  We use a \emph{quadratic statistic} built from $\Psi$, together with a
covariance bound tailored to matchings.  Split the samples into two independent blocks and compute
\[
  S=\left\langle \overline\Psi^{(1)},\overline\Psi^{(2)}\right\rangle .
\]
The independence of the two blocks makes $S$ an unbiased estimator of $\|\nu\|_2^2$; using
$\|\overline\Psi\|_2^2$ from a single block would require subtracting a noise term and would create
another bias-control problem.  To bound the variance, we use the Poincar\'e inequality for the
hard-core measure on the line graph, derived from the same spectral-independence and
entropy-factorization results~\cite{CG23,CLV23}.
The preliminary checks force $w_e\in[1/2,2]$, so the Poincar\'e constant is bounded in terms of the
degree rather than the unknown edge weights $w_e$.  The Poincar\'e bound is then combined with a
purely local sensitivity estimate: changing one edge variable can affect only the coordinates
associated with that edge and its incident neighbors, and each such coordinate has magnitude
$O_\Delta(1)$.  Hence the covariance matrix of $\Psi$ has operator norm $O_\Delta(1)$ and
squared Frobenius norm $O_\Delta(m)$.  This gives
\[
  \Var(S)\le O_\Delta\!\left(\frac{m}{s^2}+\frac{\|\nu\|_2^2}{s}\right),
\]
so $s\asymp_\Delta \sqrt m/\varepsilon^2$ suffices.  Repetition gives the stated confidence
parameter.  The resulting tester is therefore not learning the whole edge-weight vector: it
estimates $\|\E_{D_w}[\Psi]\|_2^2$.

\paragraph{Lower bound: hiding many perturbations inside a matching constraint.}
The lower bound shows that the $\sqrt m$ dependence is not an artifact of the tester.  For
bounded-degree graphs, one can already see the right order from disjoint edges, where the model
contains product Bernoulli distributions.  The main point of the lower bound is that the same lower
bound persists in dense graphs, where the matching constraint creates strong dependencies and a
sample contains only $O(n)$ occupied edges even though the graph may have $\Theta(n^2)$ edges.

The construction starts from one dense component $H_q=K_{2q,2q}$.  The left and right vertices are
partitioned into pairs, and each block between a left pair and a right pair is assigned the
\emph{checkerboard matrix}
\[
  B=\begin{pmatrix}1&-1\\ -1&1\end{pmatrix}.
\]
For a sign matrix $\sigma\in\{\pm1\}^{q\times q}$, define
$\mathsf{H}_\sigma\in\R^{E(H_q)}$ by putting the four entries of $\sigma_{uv}B$ on the four edges
in the block corresponding to $(u,v)$.  Thus $\mathsf{H}_\sigma$ is indexed by the edges of
$K_{2q,2q}$, while $\sigma$ is indexed by the $q^2$ blocks.  Perturb the all-one weights by
$w_e^{(\sigma)}=\exp(\gamma \mathsf{H}_\sigma(e))$.  This is a dense analogue of Paninski-type balanced
perturbations~\cite{Paninski08}: there are $q^2$ hidden signs, but every row and every column of $\mathsf{H}_\sigma$ sums to
zero.

The zero row and column sums remove the first-order effects that a tester could otherwise detect
from simple vertex or size statistics.  They also prevent row sums, column sums, and matching-size
fluctuations from contributing large terms in the second-moment calculation.  With these
terms absent, the \emph{checkerboard scores}
\[
  T_\sigma(M)=\sum_{e\in M} \mathsf{H}_\sigma(e)
\]
are centered and symmetric under the uniform distribution over all matchings.  These scores also
yield an isotropic covariance formula:
\[
  \Cov_U(T_\sigma,T_\tau)=4\beta_q\sum_{u,v}\sigma_{uv}\tau_{uv},
  \qquad
  \beta_q\asymp \frac1q.
\]
Setting $\tau=\sigma$ gives $\Var_U(T_\sigma)=4\beta_q q^2=\Theta(q)$.  Thus a component contains
$q^2$ random signs, but the variance of each score $T_\sigma$ matches the typical size of a
matching in $K_{2q,2q}$.

To build no-instances at total variation distance $\Omega(\varepsilon)$, take $t$
disjoint copies of $K_{2q,2q}$ and choose independent checkerboard signs
$\boldsymbol{\sigma}=(\sigma^{(1)},\dots,\sigma^{(t)})$ on the copies.  The sum of the component scores
\[
  T_{\boldsymbol{\sigma}}=\sum_{r=1}^t T_{\sigma^{(r)}}^{(r)}
\]
then has variance $\Theta(tq)$.  The log-likelihood ratio for the corresponding no-instance has
linear term $\gamma T_{\boldsymbol{\sigma}}$, whose standard deviation is therefore
$\Theta(\gamma\sqrt{tq})$.  We set
\[
  \gamma\asymp \frac{\varepsilon}{\sqrt{tq}}.
\]
With this choice of $\gamma$, a typical log-likelihood ratio fluctuates by order $\varepsilon$,
enough to make each fixed no-instance far from $U$.  Proving this rigorously requires a local
quadratic expansion of $\log \E_U e^{sT_\sigma}$ on $K_{2q,2q}$.  Since the distribution over
matchings is not a product distribution, this expansion is not automatic.  The proof exposes the
matching one row-pair at a time.  The checkerboard symmetry makes the increments a bounded
\emph{martingale-difference sequence}; then swapping the two vertices in each row pair and averaging
over these swaps turns the exponential moment into a product of $\cosh$ terms.

For indistinguishability, choose the signs uniformly at random and bound the $\chi^2$ divergence
of the $k$-sample mixture from the $k$-sample uniform distribution.  Write
$L_{\boldsymbol{\sigma}}=dD_{\boldsymbol{\sigma}}/dU$ for the likelihood ratio of one distribution
indexed by a sign collection, and let $\overline D_k$ be the $k$-sample mixture over random signs with likelihood ratio
$\overline L_k=d\overline D_k/dU^{\otimes k}$.  Then
$\E_{U^{\otimes k}}[\overline L_k]=1$, and hence
\[
  \chi^2(\overline D_k,U^{\otimes k})
  =\E_{U^{\otimes k}}[\overline L_k^2]-1.
\]
Expanding the square of the mixture likelihood ratio gives
\[
  \E_{U^{\otimes k}}[\overline L_k^2]
  =
  \E_{\boldsymbol{\sigma},\boldsymbol{\tau}}
  \left(\E_U[L_{\boldsymbol{\sigma}}L_{\boldsymbol{\tau}}]\right)^k.
\]
Thus the $\chi^2$ distance is controlled by the average value of this expression for two
independently chosen sign patterns.
For such $\boldsymbol{\sigma},\boldsymbol{\tau}$,
\[
  \log \E_U[L_{\boldsymbol{\sigma}} L_{\boldsymbol{\tau}}]
  =4\beta_q\gamma^2
    \sum_{r=1}^t\sum_{u,v}\sigma^{(r)}_{uv}\tau^{(r)}_{uv}
    +O(tq\gamma^4).
\]
The sum is a sum of $tq^2$ independent uniform $\{\pm1\}$ variables, so it has mean $0$ and
variance $tq^2$.  In the $k$th power above, the logarithm in the display is multiplied by $k$.
A standard bound for
sums of independent signs then gives
\[
  \E_{U^{\otimes k}}[\overline L_k^2]
  \le
  \exp\!\left(O(k^2\beta_q^2\gamma^4\,tq^2+k\,tq\gamma^4)\right).
\]
Here the first term comes from the random sum and the second from the error term in the display.
Substituting $\beta_q\asymp 1/q$ and
$\gamma^2\asymp \varepsilon^2/(tq)$ shows that the mixture second moment remains bounded as long as
\[
  k\lesssim \frac{q\sqrt t}{\varepsilon^2}
  =\Theta\!\left(\frac{\sqrt{|E(G_{t,q})|}}{\varepsilon^2}\right),
  \qquad |E(G_{t,q})|=4tq^2.
\]
No tester using fewer samples can reliably distinguish $U$ from this mixture, and hence cannot
solve the composite testing problem.

Finally, varying $q$ and $t$ interpolates between sparse and dense graphs.  For every pair of target
vertex and edge counts $n\le m\le n^2$, the construction yields graphs with $\Theta(n)$ vertices and
$\Theta(m)$ edges while preserving the same $\Omega(\sqrt m/\varepsilon^2)$ lower bound.  This
shows that the lower bound is not confined to sparse graphs: the same $\sqrt m$ lower bound remains
as density increases, and in the bounded-degree regime it matches the upper-bound dependence on
$m=|E|$ up to degree-dependent and logarithmic factors.

\subsection{Related work}

\paragraph{Testing high-dimensional structured distributions.}
Distribution testing over arbitrary high-dimensional domains suffers from the usual curse of
dimensionality, so a central line of work asks when structural assumptions make identity and
closeness testing possible with sample complexity polynomial in the ambient dimension rather than
in the ambient domain size, and in some cases even sublinear in the dimension.  A systematic
starting point in this direction is the work of Canonne, Diakonikolas,
Kane, and Stewart on testing Bayesian networks~\cite{CDKS20}, where the directed graphical
structure decomposes the testing problem into local conditional pieces, yielding polynomial sample
bounds for constant in-degree Bayesian networks. Concurrent work \cite{DaskalakisP17} and subsequent works \cite{canonne2026entropy,CY24} study this class of problems with similar localization ideas, but different sets of techniques. Daskalakis, Dikkala, and Kamath initiated the
analogous study for Markov random fields through independence and goodness-of-fit testing for
Ising models~\cite{DDK18}, showing that graphical structure can avoid the exponential sample
complexity of unstructured multivariate distribution testing.  Our work belongs to the same broad
program, but with a different source of structure: the domain is the constrained family
$\Mcal(G)\subseteq\{0,1\}^E$ rather than a product space, and the relevant local tests are
conditioned on which edges can be selected under the matching constraint.

\paragraph{Spin systems and graphical models.}
Identity testing for spin systems and undirected graphical models has also been studied directly.
For Ising models, Daskalakis, Dikkala, and Kamath give positive results for independence and
goodness-of-fit testing~\cite{DDK18}, while Bez\'akov\'a, Blanca, Chen, \v{S}tefankovi\v{c}, and
Vigoda prove hardness results for
antiferromagnetic Ising models and proper colorings~\cite{BBCSV20}; related hardness results for
restricted Boltzmann machines and Potts models appear in~\cite{BCSV21}.  These works place
testing in the same landscape as structure learning and approximate counting for spin systems.
The monomer-dimer model considered here is another natural graphical model, but its variables live
on the edges of $G$, with hard mutual-exclusion constraints at vertices.

\paragraph{Monomer-dimer counting and sampling.}
The monomer-dimer model is a classical model from statistical physics~\cite{HL72} and has a long
algorithmic history because matchings are a central combinatorial structure.  Markov-chain methods
for matchings and permanents go back to Jerrum and Sinclair~\cite{JS89}; for counting all
matchings in bounded-degree graphs, Bayati, Gamarnik, Katz, Nair, and Tetali gave a deterministic
FPTAS based on correlation decay~\cite{BGKNT07}.  Lelarge and Zhou studied sublinear-time
approximation of monomer-dimer partition functions and matching statistics on bounded-degree
graphs, motivated in part by massive network settings~\cite{LelargeZhou14}.  More recently, Chen,
Liu, and Vigoda developed
entropy-factorization and spectral-independence tools that imply, among other applications,
$O(m\log n)$ mixing for Glauber dynamics for the monomer-dimer model on bounded-degree graphs
with $m$ edges~\cite{CLV23}.  We use this literature both as motivation and as a source of proof
tools: bounded-degree counting gives the estimates of the probabilities $U(A_e)$ used in the
preprocessing step, while entropy and Poincar\'e inequalities supply, respectively, the lower bound
on $\|\E_{D_w}[\Psi]\|_2$ and the covariance bound used by the quadratic statistic.  The question
here is different from sampling or counting, however: given samples from the Gibbs distribution
induced by an unknown edge-weight vector $w$, we ask whether this distribution is exactly the
uniform distribution over all matchings or is far from it in total variation.

\subsection{Organization}

Section~\ref{sec:preliminaries} sets up the Gibbs family over matchings, proves the rigidity of the
all-one point, records the conditional distribution of one edge given all other edges, and states the
entropy and
Poincar\'e tools used later.  Section~\ref{sec:upper-bound} proves the upper bound.  It describes the
preprocessing step that estimates the probabilities $U(A_e)$, the screening step that reduces to
bounded edge weights, the vector $\Psi$, the covariance estimate, and the final quadratic
tester.
Section~\ref{sec:lower-bound} proves the lower bound via checkerboard perturbations on disjoint
copies of complete bipartite graphs.  Appendix~\ref{sec:learning} gives an information-theoretic
proper learning bound for the entire Gibbs family over matchings, to contrast the sample complexity of
testing with that of learning.

\section{Preliminaries}
\label{sec:preliminaries}

\subsection{Setup}

Let $G=(V,E)$ be a finite graph, and let $m:=|E|$.  Denote by
\[
\Mcal(G):=\{M\subseteq E:\ M \text{ is a matching}\}
\]
the family of all matchings of $G$, including the empty matching.
For edge weights $w=(w_e)_{e\in E}\in \R_{>0}^E$, define the Gibbs distribution
\[
D_w(M):=\frac{1}{Z(w)}\prod_{e\in M} w_e,
\qquad
Z(w):=\sum_{M\in \Mcal(G)} \prod_{e\in M} w_e.
\]
Let $U$ denote the uniform distribution on $\Mcal(G)$.

We study the promise problem
\[
H_0:\ D_w=U
\qquad\text{versus}\qquad
H_1:\ \TV(D_w,U)\ge \varepsilon.
\]

\subsection{Uniformity is rigid in the all-matchings model}

\begin{proposition}[Uniformity characterization]
\label{prop:uniformity-characterization}
For $w\in \R_{>0}^E$,
\[
D_w=U
\quad\Longleftrightarrow\quad
w_e=1 \ \text{for all } e\in E.
\]
\end{proposition}

\begin{proof}
The empty matching $\varnothing$ has weight $1$, whereas the singleton matching $\{e\}$ has
weight $w_e$.  If $D_w$ is uniform over all matchings, then $\varnothing$ and $\{e\}$ must have the
same probability for every edge $e$, hence $w_e=1$ for every $e$.  The converse is immediate.
\end{proof}

Thus testing uniformity in this family is the same as testing whether $w\equiv 1$.

\subsection{Conditioning on all other edges}

Fix an edge $e\in E$.  For a matching $M$, define
\[
X_e(M):=\one\{e\in M\},
\]
and also
\[
A_e(M):=\one\{\text{no edge adjacent to $e$ belongs to } M\setminus\{e\}\}.
\]
So $A_e(M)=1$ means that the edge $e$ is \emph{available} given the rest of the configuration.
We also write
\[
\xi_e(M):=A_e(M)(2X_e(M)-1).
\]
Thus $\xi_e(M)=1$ if $e\in M$, $\xi_e(M)=-1$ if $e$ is available but not in $M$, and
$\xi_e(M)=0$ if $e$ is not available.
When the underlying matching is random, we suppress the argument and write $X_e$, $A_e$, and
$\xi_e$ for the corresponding random variables.

\begin{lemma}[Conditional distribution of one edge]
\label{lem:one-edge-conditional-distribution}
For every edge $e\in E$ and every configuration of the other edge indicators,
\[
D_w(X_e=1\mid X_{-e}) = A_e\,\frac{w_e}{1+w_e},
\qquad
U(X_e=1\mid X_{-e}) = A_e\,\frac12.
\]
Equivalently, with
\[
p_e:=\frac{w_e}{1+w_e},
\qquad
t_e:=2p_e-1=\frac{w_e-1}{w_e+1},
\]
the conditional distribution on the event $A_e=1$ is
\[
\Ber(p_e)
\quad\text{under } D_w,
\qquad
\Ber(1/2)
\quad\text{under } U.
\]
\end{lemma}

\begin{proof}
Once all other edges are fixed, there are only two possibilities for $X_e$.
Either some adjacent edge is already present, in which case $X_e=0$ is forced, or else $e$ may be
present or absent.  In the latter case the Gibbs weights of the two feasible options differ by the
factor $w_e:1$, while under $U$ the two feasible options are equally likely.
\end{proof}

\subsection{Entropy and Poincare tools on the line graph}

From now on we assume that the maximum degree of $G$ is bounded by a fixed constant:
\[
\Delta(G)\le \Delta.
\]
For a probability measure $\pi$ and density $f\ge 0$, write
\[
\Ent_\pi(f):=\E_\pi[f\log f]-\E_\pi[f]\log\E_\pi[f].
\]
For an edge $e$, let $\Ent_e(f)$ and $\Var_e(g)$ denote the conditional entropy and conditional
variance in the coordinate $X_e$, with all other coordinates held fixed.

\begin{remark}[Line-graph representation]
\label{rem:line-graph}
Let $\Lcal(G)$ be the line graph of $G$.  A subset $M\subseteq E$ is a matching of $G$ if and only
if it is an independent set of $\Lcal(G)$.  Hence $D_w$ is exactly the hard-core measure on
$\Lcal(G)$ with vertex weights $\lambda_e=w_e$.  Moreover $\Lcal(G)$ is claw-free and has
maximum degree at most $2\Delta-2$.
\end{remark}

The next proposition states the imported inequalities with an explicit dependence on the maximum
degree.  The constants are intentionally not optimized; they are obtained by inserting the explicit
spectral-independence bound for claw-free hard-core models into the explicit entropy-factorization
constant of Chen, Liu, and Vigoda.

For the rest of this subsection, set
\begin{align}
\label{eq:explicit-imported-constants}
\widehat\Delta&:=\max\{3,2\Delta-2\},\\
\eta_\Delta&:=4\widehat\Delta+2,\\
b_\Delta^{\mathrm{hc}}&:=\frac{1}{3(2\widehat\Delta^2+1)},\\
C_\Delta^{\mathrm{ent}}
&:=
\frac{18\log(1/b_\Delta^{\mathrm{hc}})}{(b_\Delta^{\mathrm{hc}})^4}
\left(\frac{24\widehat\Delta}{(b_\Delta^{\mathrm{hc}})^2}\right)^{
4\eta_\Delta/(b_\Delta^{\mathrm{hc}})^2+1}.
\label{eq:cent-definition}
\end{align}
In particular, for a universal constant $C<\infty$,
\begin{equation}
\label{eq:cent-asymptotic}
C_\Delta^{\mathrm{ent}}\le (C\Delta)^{C\Delta^5}
\qquad(\Delta\ge2).
\end{equation}

\begin{proposition}[Explicit imported entropy and Poincare inequalities]
\label{prop:imported-tools}
Fix $\Delta\in\mathbb N$.  With $C_\Delta^{\mathrm{ent}}$ defined in
\eqref{eq:cent-definition}, the following holds for every graph $G=(V,E)$ with
$\Delta(G)\le \Delta$.

\begin{enumerate}
\item The uniform distribution $U$ on $\Mcal(G)$ satisfies the approximate tensorization of entropy:
for every density $f=d\nu/dU$,
  \[
  \Ent_U(f)
  \le
  C_\Delta^{\mathrm{ent}}
  \sum_{e\in E} \E_U\!\left[\Ent_e(f)\right].
  \]
Equivalently, for every Gibbs distribution $D_w$,
\[
\KL(D_w\|U)
\le
C_\Delta^{\mathrm{ent}}
\sum_{e\in E}
\E_{D_w}\!\left[
\KL\!\bigl(D_w(X_e\mid X_{-e})\,\|\,U(X_e\mid X_{-e})\bigr)
\right].
\]

\item If, in addition,
\[
\frac12\le w_e\le 2
\qquad\text{for all } e\in E,
\]
then $D_w$ satisfies the Poincare inequality
\[
\Var_{D_w}(g)
\le
C_\Delta^{\mathrm{ent}}
\sum_{e\in E}
\E_{D_w}\!\left[\Var_e(g)\right]
\]
for every $g:\Mcal(G)\to\R$.
\end{enumerate}
\end{proposition}

\begin{proof}
By Remark~\ref{rem:line-graph}, every Gibbs distribution over matchings of $G$ is a hard-core measure on
$\Lcal(G)$, and $\Lcal(G)$ is a claw-free graph of maximum degree at most $\widehat\Delta$.

We record the compatibility checks needed for the two imported results.  Consider any hard-core
measure $\mu$ on a claw-free graph $H$ of maximum degree at most $\widehat\Delta$, with vertex
weights $\lambda_v\in[1/2,2]$; write $\lambda_{\min}$ and $\lambda_{\max}$ for their minimum
and maximum.
\begin{enumerate}[label=(\alph*)]
\item \emph{Pinnings.}
Conditioning a hard-core measure on a feasible pinning gives another hard-core measure on the
induced graph obtained by deleting pinned vertices and the neighbors of vertices pinned to be
occupied.  This induced graph remains claw-free and has maximum degree at most
$\widehat\Delta$, and the remaining weights are still in $[1/2,2]$.

\item \emph{Spectral independence.}
Chen and Gu's hard-core theorem for claw-free graphs~\cite[Theorem~4]{CG23} gives
$2(1+\Delta_H\lambda_{\max})$-spectral independence for a claw-free hard-core instance of maximum
degree $\Delta_H$.  Applying this after every feasible pinning and using
$\Delta_H\le\widehat\Delta$ and $\lambda_{\max}\le2$ gives the uniform bound
\[
2\bigl(1+\widehat\Delta\lambda_{\max}\bigr)
\le 2(1+2\widehat\Delta)=\eta_\Delta.
\]

\item \emph{Marginal boundedness.}
For an unpinned vertex $v$, spin $0$ has marginal probability at least
$1/(1+\lambda_{\max})\ge 1/3$.  If spin $1$ is feasible, then conditioning on all vertices outside
$N(v)\cup\{v\}$, the neighbors of $v$ form a claw-free neighborhood and hence contain no
independent set of size three.  Therefore the conditional weight of feasible configurations on
$N(v)$ is at most
\[
1+\widehat\Delta\lambda_{\max}
  +\binom{\widehat\Delta}{2}\lambda_{\max}^2
\le 1+2\widehat\Delta+2\widehat\Delta(\widehat\Delta-1)
=2\widehat\Delta^2+1
\]
times the weight of the all-zero configuration on $N(v)$.  Thus, whenever spin $1$ is feasible,
\[
\mu(\sigma_v=1)
\ge
\frac{\lambda_{\min}}{1+\lambda_{\min}}\cdot \frac{1}{2\widehat\Delta^2+1}
\ge
\frac{1}{3(2\widehat\Delta^2+1)}
=b_\Delta^{\mathrm{hc}}.
\]
The same argument applies after every feasible pinning by item~(a), so $\mu$ is
$b_\Delta^{\mathrm{hc}}$-marginally bounded.

\item \emph{Total connectedness.}
Under any pinning, feasible configurations are independent sets of the residual induced graph.
Deleting occupied vertices one at a time connects every feasible configuration to the empty
configuration, so the support is connected under single-site moves.

\item \emph{Small size.}
Chen, Liu, and Vigoda's approximate-tensorization theorem~\cite[Theorem~2.9]{CLV23} also assumes
the number of vertices in the spin-system graph is at least
\[
\frac{24\widehat\Delta}{(b_\Delta^{\mathrm{hc}})^2}
\left(\frac{4\eta_\Delta}{(b_\Delta^{\mathrm{hc}})^2}+1\right).
\]
If $H$ has fewer vertices, add isolated vertices with weights in $[1/2,2]$ until this condition
holds.  The augmented graph is still claw-free, has maximum degree at most $\widehat\Delta$, and
the augmented measure is the product of $\mu$ and independent one-site hard-core measures.
Applying the theorem to functions that depend only on the original coordinates gives the same
constant for $\mu$.
\end{enumerate}

Therefore \cite[Theorem~2.9]{CLV23}, with maximum degree $\widehat\Delta$, marginal parameter
$b=b_\Delta^{\mathrm{hc}}$, and spectral-independence parameter $\eta_\Delta$, yields approximate
tensorization of entropy with constant
\[
\frac{18\log(1/b_\Delta^{\mathrm{hc}})}{(b_\Delta^{\mathrm{hc}})^4}
\left(\frac{24\widehat\Delta}{(b_\Delta^{\mathrm{hc}})^2}\right)^{
4\eta_\Delta/(b_\Delta^{\mathrm{hc}})^2+1}
=C_\Delta^{\mathrm{ent}}.
\]

Part~(1) applies this conclusion to the special all-one weight vector $\lambda_e\equiv 1$, which is
exactly the uniform distribution $U$ on matchings.  The displayed KL form is the same entropy
tensorization with $f=dD_w/dU$; this part does not require the weights $w_e$ to be bounded.

Part~(2) applies the same conclusion to the hard-core measure corresponding to $D_w$, and then
uses Fact~3.5 of~\cite{CLV23}: approximate tensorization with constant
$C_\Delta^{\mathrm{ent}}$ yields spectral gap at least $1/(C_\Delta^{\mathrm{ent}}m)$ for the
heat-bath Glauber dynamics on the $m$ vertices of $\Lcal(G)$.  Since the Dirichlet form of this
dynamics is
\[
\frac1m\sum_{e\in E}\E_{D_w}[\Var_e(g)],
\]
the stated Poincare inequality follows.  Finally, for $\Delta\ge2$ one has
$\widehat\Delta\le2\Delta$, $(b_\Delta^{\mathrm{hc}})^{-1}\le C\Delta^2$, and
$\eta_\Delta\le C\Delta$, which gives~\eqref{eq:cent-asymptotic}.
\end{proof}

\subsection{A lower bound on edge availability under \texorpdfstring{$U$}{U}}

\begin{lemma}[Availability is bounded away from zero under $U$]
\label{lem:availability-lower-bound}
For every edge $e\in E$,
\[
\alpha_e:=U(A_e)\ge \frac{1}{\Delta^2}.
\]
\end{lemma}

\begin{proof}
Fix an edge $e=uv$.  Let
\[
\mathcal A_e:=\{M\in \Mcal(G): A_e(M)=1\},
\qquad
\mathcal B_e:=\{M\in \Mcal(G): A_e(M)=0\}.
\]
For $M\in \mathcal B_e$, delete from $M$ all edges incident to $u$ or $v$ and call the resulting
matching $R(M)$.  Then $R(M)\in \mathcal A_e$.

Conversely, fix $M'\in \mathcal A_e$.  To recover a preimage $M$ of $M'$ under $R$, we may add back
at most one edge incident to $u$ (other than $e$) and at most one edge incident to $v$ (other than
$e$).  Hence the number of choices is at most
\[
(\deg(u)-1)+(\deg(v)-1)+(\deg(u)-1)(\deg(v)-1)\le \Delta^2-1.
\]
Therefore
\[
|\mathcal B_e|\le (\Delta^2-1)|\mathcal A_e|.
\]
It follows that
\[
\alpha_e
= \frac{|\mathcal A_e|}{|\mathcal A_e|+|\mathcal B_e|}
\ge \frac{1}{\Delta^2}.
\qedhere
\]
\end{proof}

\section{Upper Bound}
\label{sec:upper-bound}

The tester below is written as an explicit algorithm.  Its only preprocessing step is to approximate,
under $U$, the probability that each edge is available.  After this preprocessing, the
quadratic statistic uses the normalized variables
$\psi_e=\xi_e/\sqrt{\widetilde\alpha_e}$ defined below.
The screening step separately handles large shifts in availability probabilities.

\subsection{Preprocessing: estimating \texorpdfstring{$U(A_e)$}{U(Ae)}}

For an edge $e\in E$, let
\[
\Gamma(e):=\{f\in E\setminus\{e\}: f\text{ shares an endpoint with }e\},
\]
and let
\[
G^{(e)}:=(V,E\setminus \Gamma(e)).
\]
Thus $e$ itself is retained, while the edges conflicting with $e$ are removed.  Recall
\[
\alpha_e:=U(A_e).
\]

\begin{lemma}[Computing estimates of \texorpdfstring{$U(A_e)$}{U(Ae)}]
\label{lem:compute-alpha}
Fix an integer $\Delta\ge 1$ and set
\[
\tau:=\frac{1}{4\Delta^2}.
\]
For every graph $G=(V,E)$ with $\Delta(G)\le \Delta$, one can compute numbers
$(\widetilde\alpha_e)_{e\in E}$ satisfying
\[
|\widetilde\alpha_e-\alpha_e|\le \frac{\tau}{8}
\qquad(e\in E)
\]
in deterministic polynomial time for every fixed $\Delta$.  Moreover,
\[
\frac{31}{32}\alpha_e\le \widetilde\alpha_e\le \frac{33}{32}\alpha_e,
\qquad
\widetilde\alpha_e\ge \frac{31}{32\Delta^2}.
\]
\end{lemma}

\begin{proof}
The event $A_e=1$ is exactly the event that the sampled matching uses no edge in $\Gamma(e)$.
Equivalently, it is the event that the matching belongs to $\Mcal(G^{(e)})$.  Therefore
\begin{equation}
\label{eq:alpha-counting-ratio}
\alpha_e=\frac{|\Mcal(G^{(e)})|}{|\Mcal(G)|}.
\end{equation}
The graphs $G$ and $G^{(e)}$ all have maximum degree at most $\Delta$.
Use a deterministic FPTAS for counting all matchings in bounded-degree graphs, for instance
Bayati, Gamarnik, Katz, Nair, and Tetali~\cite{BGKNT07}, to approximate the $m+1$ quantities
$|\Mcal(G)|$ and $|\Mcal(G^{(e)})|$, $e\in E$, to a sufficiently small relative accuracy
$\rho_\Delta>0$.  Taking the ratio of the two estimates in~\eqref{eq:alpha-counting-ratio} gives
$\widetilde\alpha_e$.  Since $\rho_\Delta$ is a constant depending only on $\Delta$, the total
preprocessing time is polynomial in $m$ for fixed $\Delta$.

Choosing $\rho_\Delta$ sufficiently small ensures
$|\widetilde\alpha_e-\alpha_e|\le \tau/8$ for every $e$.  The remaining two inequalities follow from
Lemma~\ref{lem:availability-lower-bound}, since
\[
\frac{\tau}{8}=\frac{1}{32\Delta^2}\le \frac{\alpha_e}{32}.
\]
\end{proof}

\subsection{Screening to bound the weights}

Throughout the rest of the upper-bound proof, let $(\widetilde\alpha_e)_{e\in E}$ be the output of
Lemma~\ref{lem:compute-alpha}.  Given $N_0$ i.i.d. samples
$M_1,\dots,M_{N_0}\sim D_w$, compute
\[
\widehat a_e:=\frac1{N_0}\sum_{i=1}^{N_0} A_e(M_i),
\qquad
\widehat{\xi}_e:=\frac1{N_0}\sum_{i=1}^{N_0} \xi_e(M_i).
\]
Reject if, for some edge $e$,
\[
|\widehat a_e-\widetilde\alpha_e|>\frac{\tau}{2}
\qquad\text{or}\qquad
|\widehat{\xi}_e|>\frac{\tau}{2}.
\]

\begin{lemma}[Consequences of the screening step]
\label{lem:screening-consequences}
Assume $m\ge 1$ and $0<\delta<1$.  If
\[
N_0\ge 2048\Delta^4\log\frac{4m}{\delta},
\]
then with probability at least $1-\delta$ the following event holds:

\begin{enumerate}
\item under $H_0$, the screening step does not reject;
\item if the screening step accepts, then for every edge $e$,
\[
|a_e-\alpha_e|\le \frac{3\tau}{4},
\qquad
|a_e t_e|\le \frac{5\tau}{8},
\qquad
\frac12\le w_e\le 2,
\]
where $a_e:=D_w(A_e)$ and $t_e=(w_e-1)/(w_e+1)$;
\item if the screening step accepts, then
\[
\frac12 a_e\le \widetilde\alpha_e\le 2a_e
\qquad(e\in E).
\]
\end{enumerate}
\end{lemma}

\begin{proof}
Since each summand in $\widehat a_e$ and $\widehat{\xi}_e$ lies in $[-1,1]$, Hoeffding's inequality and
a union bound over the $2m$ empirical quantities imply that, with the displayed value of $N_0$,
with probability at least $1-\delta$,
\begin{equation}
\label{eq:screening-concentration}
|\widehat a_e-a_e|\le \frac{\tau}{8},
\qquad
|\widehat{\xi}_e-a_et_e|\le \frac{\tau}{8}
\qquad(\forall e\in E).
\end{equation}
Work on this event.

Under $H_0$, Proposition~\ref{prop:uniformity-characterization} gives $w_e=1$ for all $e$.
Hence $t_e=0$ and $a_e=\alpha_e$.  Therefore, using Lemma~\ref{lem:compute-alpha},
\[
|\widehat a_e-\widetilde\alpha_e|
\le |\widehat a_e-\alpha_e|+|\alpha_e-\widetilde\alpha_e|
\le \frac{\tau}{8}+\frac{\tau}{8}
<\frac{\tau}{2},
\]
and
\[
|\widehat{\xi}_e|\le \frac{\tau}{8}<\frac{\tau}{2}.
\]
Thus screening does not reject under $H_0$ on the concentration event.

Now suppose that screening accepts.  Combining acceptance, Lemma~\ref{lem:compute-alpha}, and
\eqref{eq:screening-concentration} gives, for every edge $e$,
\[
|a_e-\alpha_e|
\le
|a_e-\widehat a_e|+|\widehat a_e-\widetilde\alpha_e|+|\widetilde\alpha_e-\alpha_e|
\le
\frac{\tau}{8}+\frac{\tau}{2}+\frac{\tau}{8}
=\frac{3\tau}{4},
\]
and
\[
|a_e t_e|
\le
|a_et_e-\widehat{\xi}_e|+|\widehat{\xi}_e|
\le
\frac{\tau}{8}+\frac{\tau}{2}
=\frac{5\tau}{8}.
\]
Since $\tau=1/(4\Delta^2)$ and $\alpha_e\ge \Delta^{-2}=4\tau$ by
Lemma~\ref{lem:availability-lower-bound},
\[
a_e\ge \alpha_e-\frac{3\tau}{4}
\ge \frac{13}{16}\alpha_e
\ge \frac{13}{16\Delta^2}.
\]
Consequently,
\[
|t_e|
\le
\frac{5\tau/8}{13/(16\Delta^2)}
=\frac{10}{52}<\frac13.
\]
Since $w_e=(1+t_e)/(1-t_e)$, this implies $1/2\le w_e\le 2$.

Finally, the estimates
\[
\frac{13}{16}\alpha_e\le a_e\le \frac{19}{16}\alpha_e,
\qquad
\frac{31}{32}\alpha_e\le \widetilde\alpha_e\le \frac{33}{32}\alpha_e
\]
imply $a_e/2\le \widetilde\alpha_e\le 2a_e$ after loosening constants.
\end{proof}

\subsection{The vector \texorpdfstring{$\Psi$}{Psi} and its mean}

For each edge, define the normalized statistic
\[
\psi_e(M):=\frac{\xi_e(M)}{\sqrt{\widetilde\alpha_e}}
\qquad(e\in E),
\]
and concatenate these coordinates into
\begin{equation}
\label{eq:normalized-statistics}
\Psi(M):=(\psi_e(M))_{e\in E}\in\R^m.
\end{equation}
Write
\[
\nu:=\E_{D_w}[\Psi],
\qquad
\nu_e=\frac{a_e t_e}{\sqrt{\widetilde\alpha_e}}.
\]
Also define
\[
h(t):=
\KL\!\left(\Ber\!\left(\frac{1+t}{2}\right)\middle\|\Ber\!\left(\frac12\right)\right)
=
\frac12\Bigl[(1+t)\log(1+t)+(1-t)\log(1-t)\Bigr].
\]

\begin{lemma}[Mean lower bound after screening]
\label{lem:mean-lower-bound}
Assume that the screening step accepts on the event of Lemma~\ref{lem:screening-consequences}.  Then, with
\[
\kappa_\Delta:=\frac{1}{C_\Delta^{\mathrm{ent}}\log 2},
\]
one has
\[
\TV(D_w,U)\ge \varepsilon
\qquad\Longrightarrow\qquad
\|\nu\|_2^2\ge \kappa_\Delta\varepsilon^2.
\]
Under $H_0$, one has $\nu=0$.
\end{lemma}

\begin{proof}
By Proposition~\ref{prop:imported-tools}(1) and Lemma~\ref{lem:one-edge-conditional-distribution},
\begin{equation}
\label{eq:kl-by-local-conditional}
\KL(D_w\|U)
\le
C_\Delta^{\mathrm{ent}}
\sum_{e\in E} a_e h(t_e).
\end{equation}
For $|t|\le 1$,
\begin{equation}
\label{eq:h-quadratic-bound}
h(t)\le (\log 2)t^2,
\end{equation}
because
\[
h(t)=\sum_{k\ge1}\frac{t^{2k}}{(2k-1)(2k)}
\le
\left(\sum_{k\ge1}\frac{1}{(2k-1)(2k)}\right)t^2
=(\log 2)t^2.
\]
On the accepted screening event, Lemma~\ref{lem:screening-consequences} gives $|t_e|\le 1/3$ and
$\widetilde\alpha_e\le 2a_e$.  Hence
\[
a_e h(t_e)
\le
(\log 2)a_e t_e^2
\le
2\log 2\,\frac{a_e^2t_e^2}{\widetilde\alpha_e}
=
2\log 2\,\nu_e^2.
\]
Combining this with~\eqref{eq:kl-by-local-conditional} yields
\[
\KL(D_w\|U)
\le
2C_\Delta^{\mathrm{ent}}\log 2\,\|\nu\|_2^2.
\]
Pinsker's inequality gives $\KL(D_w\|U)\ge 2\TV(D_w,U)^2$, and therefore
$\|\nu\|_2^2\ge \kappa_\Delta\varepsilon^2$ for
$\kappa_\Delta=(C_\Delta^{\mathrm{ent}}\log 2)^{-1}$.

Under $H_0$, conditional on $X_{-e}$, either $A_e=0$ or $A_e=1$ and
$X_e\sim\Ber(1/2)$ by Lemma~\ref{lem:one-edge-conditional-distribution}.  Therefore
$\E_U[\xi_e]=0$ for every $e$, so $\nu=0$.
\end{proof}

\subsection{Covariance bound when weights are bounded}

Assume now that
\[
\frac12\le w_e\le 2
\qquad\text{for all } e\in E.
\]
Write any vector $c\in\R^m$ as $c=(c_f)_{f\in E}$, and define
\[
g_c(M):=\langle c,\Psi(M)\rangle=\sum_{f\in E}c_f\psi_f(M).
\]
For an edge $e$, let
\[
B_e:=\{f\in E:\ f=e \text{ or } f\sim e\}
\]
be the set consisting of $e$ and all edges adjacent to $e$.  Since $\Delta(G)\le \Delta$,
\[
|B_e|\le 2\Delta-1.
\]

\begin{lemma}[Local sensitivity of \texorpdfstring{$\Psi$}{Psi}]
\label{lem:local-sensitivity}
For every $c\in\R^m$,
\[
\sum_{e\in E}\E_{D_w}\!\left[\Var_e(g_c)\right]
\le
B_\Delta\|c\|_2^2,
\qquad
B_\Delta:=4\Delta^2(2\Delta-1)^2.
\]
\end{lemma}

\begin{proof}
Fix an edge $e\in E$ and condition on $X_{-e}$.  If the conditional distribution of $X_e$ is degenerate,
then $\Var_e(g_c)=0$.  Otherwise there are exactly two feasible completions, say $M^{(0)}$ and
$M^{(1)}$, differing only in the value of $X_e$, and hence
\[
\Var_e(g_c)
\le
\frac14\bigl(g_c(M^{(1)})-g_c(M^{(0)})\bigr)^2.
\]

The coordinate $\psi_f$ depends only on $X_f$ and on the indicators of edges adjacent to $f$.
Thus changing $X_e$ can affect $\psi_f$ only when $f\in B_e$.  Moreover, since
$\xi_f\in\{-1,0,1\}$ and Lemma~\ref{lem:compute-alpha} gives
$\widetilde\alpha_f\ge 31/(32\Delta^2)$,
\[
|\psi_f(M^{(1)})-\psi_f(M^{(0)})|
\le
\frac{2}{\sqrt{\widetilde\alpha_f}}
\le
4\Delta.
\]
Therefore
\[
|g_c(M^{(1)})-g_c(M^{(0)})|
\le
4\Delta\sum_{f\in B_e}|c_f|
\le
4\Delta\sqrt{|B_e|}\left(\sum_{f\in B_e}c_f^2\right)^{1/2},
\]
using Cauchy–Schwarz, and hence
\[
\Var_e(g_c)
\le
4\Delta^2(2\Delta-1)\sum_{f\in B_e}c_f^2.
\]
Taking expectation and summing over $e$ gives
\[
\sum_{e\in E}\E_{D_w}[\Var_e(g_c)]
\le
4\Delta^2(2\Delta-1)
\sum_{e\in E}\sum_{f\in B_e}c_f^2.
\]
Each fixed $f$ belongs to $B_e$ for at most $2\Delta-1$ choices of $e$, so
\[
\sum_{e\in E}\E_{D_w}[\Var_e(g_c)]
\le
4\Delta^2(2\Delta-1)^2\sum_{f\in E}c_f^2.
\]
This proves the claim.
\end{proof}

\begin{proposition}[Local covariance bound]
\label{prop:local-covariance}
Whenever $1/2\le w_e\le 2$ for all $e\in E$, one has, for every $c\in\R^m$,
\[
\Var_{D_w}\bigl(\langle c,\Psi(M)\rangle\bigr)
\le
V_\Delta\|c\|_2^2,
\qquad
V_\Delta:=C_\Delta^{\mathrm{ent}}B_\Delta.
\]
In particular, if $C_w:=\Cov_{D_w}(\Psi)$, then
\[
\|C_w\|_{\mathrm{op}}\le V_\Delta,
\qquad
\mathrm{tr}(C_w^2)\le mV_\Delta^2.
\]
\end{proposition}

\begin{proof}
By Proposition~\ref{prop:imported-tools}(2),
\[
\Var_{D_w}(g_c)
\le
C_\Delta^{\mathrm{ent}}
\sum_{e\in E}\E_{D_w}[\Var_e(g_c)].
\]
Combining this with Lemma~\ref{lem:local-sensitivity} yields the asserted bound with
$V_\Delta:=C_\Delta^{\mathrm{ent}}B_\Delta$.
The operator-norm claim follows from the variational characterization of covariance.  Finally,
$C_w$ is a positive semidefinite $m\times m$ matrix, so
\[
\mathrm{tr}(C_w^2)\le m\|C_w\|_{\mathrm{op}}^2\le mV_\Delta^2.
\]
\end{proof}

\subsection{Quadratic statistic when weights are bounded}

Assume from now on that the screening step has accepted on the event of
Lemma~\ref{lem:screening-consequences}, so that $1/2\le w_e\le 2$ for every edge $e$.
Set
\[
\nu:=\E_{D_w}[\Psi]\in\R^m,
\qquad
C_w:=\Cov_{D_w}(\Psi).
\]
By Proposition~\ref{prop:local-covariance},
\[
\|C_w\|_{\mathrm{op}}\le V_\Delta,
\qquad
\mathrm{tr}(C_w^2)\le mV_\Delta^2.
\]

Take two independent fresh blocks of size $s$:
\[
M_1^{(1)},\dots,M_s^{(1)}\sim D_w,
\qquad
M_1^{(2)},\dots,M_s^{(2)}\sim D_w,
\]
and define
\[
\overline\Psi^{(1)}:=\frac1s\sum_{i=1}^s\Psi(M_i^{(1)}),
\qquad
\overline\Psi^{(2)}:=\frac1s\sum_{i=1}^s\Psi(M_i^{(2)}),
\qquad
S:=\left\langle \overline\Psi^{(1)},\overline\Psi^{(2)}\right\rangle.
\]

\begin{lemma}[Mean and variance of the quadratic statistic]
\label{lem:quadratic-statistic}
One has
\[
\E_{D_w}[S]=\|\nu\|_2^2,
\]
and
\[
\Var_{D_w}(S)
\le
\frac{\mathrm{tr}(C_w^2)}{s^2}
+
\frac{2\|C_w\|_{\mathrm{op}}\|\nu\|_2^2}{s}
\le
\frac{mV_\Delta^2}{s^2}
+
\frac{2V_\Delta\|\nu\|_2^2}{s}.
\]
\end{lemma}

\begin{proof}
The two blocks are independent, each with mean $\nu$, so
\[
\E[S]=\langle \nu,\nu\rangle=\|\nu\|_2^2.
\]
Also
\[
\E\bigl[\overline\Psi^{(j)}\overline\Psi^{(j)\top}\bigr]
=
\frac{1}{s}C_w+\nu\nu^\top
\qquad(j=1,2).
\]
Using independence of the two blocks,
\[
\E[S^2]
=
\mathrm{tr}\!\left(\left(\frac{1}{s}C_w+\nu\nu^\top\right)^2\right)
=
\frac{\mathrm{tr}(C_w^2)}{s^2}
+
\frac{2\nu^\top C_w\nu}{s}
+
\|\nu\|_2^4.
\]
Subtracting $\|\nu\|_2^4$ and using
$\nu^\top C_w\nu\le \|C_w\|_{\mathrm{op}}\|\nu\|_2^2$ gives the claim.
\end{proof}

\subsection{Proof of the main upper bound}

\begin{proof}[Proof of Theorem~\ref{thm:main}]
Fix an integer $\Delta\ge 1$ and define
\[
B_\Delta:=4\Delta^2(2\Delta-1)^2,
\qquad
F_\Delta:=10^5\left(\Delta^4+B_\Delta\,(C_\Delta^{\mathrm{ent}})^2\log 2\right),
\]
where $C_\Delta^{\mathrm{ent}}$ is the explicit constant in
\eqref{eq:cent-definition}.  Then
\[
F_\Delta\le (C\Delta)^{C\Delta^5}
\]
for a universal constant $C<\infty$ by~\eqref{eq:cent-asymptotic}, after increasing $C$ if needed to
cover $\Delta=1$.  Now fix a graph $G=(V,E)$ with $\Delta(G)\le \Delta$ and $m=|E|\ge 1$, and fix
$0<\varepsilon\le 1$ and $0<\delta<1$.
By Proposition~\ref{prop:uniformity-characterization}, testing uniformity in this family is exactly
testing whether $w\equiv 1$.

\medskip
\noindent
\emph{Step 0: preprocessing.}
Compute the estimates $(\widetilde\alpha_e)_{e\in E}$ from Lemma~\ref{lem:compute-alpha}.
This step uses no samples from the unknown distribution.

\medskip
\noindent
\emph{Step 1: screening.}
Take
\[
N_0=\left\lceil 2048\Delta^4\log\frac{12m}{\delta}\right\rceil
\]
samples and run the screening procedure from the previous subsection, with failure parameter
$\delta/3$.  With probability at least $1-\delta/3$, Lemma~\ref{lem:screening-consequences} ensures
that screening does not reject under $H_0$, while if screening accepts then $w_e\in[1/2,2]$ for
every edge $e$ and the comparability estimates needed below hold.  If screening rejects, output
$H_1$.

From now on, condition on the event that screening accepts and the conclusion of
Lemma~\ref{lem:screening-consequences} holds.

\medskip
\noindent
\emph{Step 2: mean lower bound.}
By Lemma~\ref{lem:mean-lower-bound}, under $H_1$ one has
\[
\|\nu\|_2^2=\|\E_{D_w}[\Psi]\|_2^2\ge \kappa_\Delta\varepsilon^2,
\qquad
\kappa_\Delta=\frac{1}{C_\Delta^{\mathrm{ent}}\log 2}.
\]
Under $H_0$ we have $\nu=0$.

\medskip
\noindent
\emph{Step 3: one quadratic block.}
Take two fresh independent blocks of size
\[
s:=\left\lceil 400\frac{V_\Delta}{\kappa_\Delta}\frac{\sqrt m}{\varepsilon^2}\right\rceil
\]
and compute the statistic $S$ from Lemma~\ref{lem:quadratic-statistic}.  This explicit choice of
$s$ gives the following two bounds.

Under $H_0$, since $\nu=0$,
\[
\Var(S)\le \frac{mV_\Delta^2}{s^2}
\le \frac{\kappa_\Delta^2\varepsilon^4}{160000}.
\]

Under $H_1$, since $\|\nu\|_2^2\ge \kappa_\Delta\varepsilon^2$ and $m\ge 1$,
\[
\frac{mV_\Delta^2}{s^2}\le \frac{\|\nu\|_2^4}{160000},
\qquad
\frac{2V_\Delta\|\nu\|_2^2}{s}
\le
\frac{\|\nu\|_2^4}{200}.
\]
Therefore
\[
\Var(S)
\le
\frac{mV_\Delta^2}{s^2}+\frac{2V_\Delta\|\nu\|_2^2}{s}
\le
\frac{\|\nu\|_2^4}{190}.
\]

Hence, by Chebyshev's inequality, under $H_0$ the first event below fails with probability at most
$9/160000$, while under $H_1$ the second event below fails with probability at most $9/190$:
\[
H_0 \implies S\le \frac{\kappa_\Delta}{3}\varepsilon^2,
\qquad
H_1 \implies S\ge \frac{2\kappa_\Delta}{3}\varepsilon^2.
\]
Thus the block test which votes for $H_1$ iff
$S>\kappa_\Delta\varepsilon^2/2$ succeeds with probability at least $0.9$.

\medskip
\noindent
\emph{Step 4: confidence boosting.}
Repeat the quadratic-block test independently
\[
L=\left\lceil 20\log(2/\delta)\right\rceil
\]
times and take a majority vote.  Standard Chernoff amplification reduces the error probability of
this stage to at most $\delta/3$.

It remains only to check the displayed sample count.  Since
\[
\frac{V_\Delta}{\kappa_\Delta}
=B_\Delta\,(C_\Delta^{\mathrm{ent}})^2\log 2,
\]
and since $m\ge1$, $0<\varepsilon\le1$, and $\log(2/\delta)\ge\log2$, the ceiling terms are absorbed by
the numerical factor $10^5$ in the definition of $F_\Delta$.  Also
\[
\log\frac{12m}{\delta}\le 2\log\frac{4m}{\delta}
\qquad(m\ge1,\ 0<\delta<1).
\]
Combining the screening stage with the $2sL$ samples used in the quadratic stage proves the
displayed sample bound.

For the running time, the preprocessing consists of $m+1$ calls to a bounded-degree matching-counting
FPTAS with constant relative accuracy depending only on $\Delta$.  Given a sample matching, all
coordinates $A_e$ and $\xi_e$ can be computed in $O_\Delta(m)$ time by scanning the
constant-size neighborhood of each edge.  Thus the preprocessing takes $m^{O_\Delta(1)}$ time, and
processing the displayed number of samples takes
\[
O_\Delta\!\left(
mF_\Delta\left(
\log\frac{4m}{\delta}
+
\frac{\sqrt m\,\log(2/\delta)}{\varepsilon^2}
\right)
\right)
\]
time.
\end{proof}

\begin{remark}
The proof shows the qualitative separation
\[
\text{learning sample complexity } \asymp \frac{m}{\varepsilon^2}
\qquad\text{versus}\qquad
\text{testing sample complexity } \asymp_\Delta \frac{\sqrt m}{\varepsilon^2}
\]
for the bounded-degree Gibbs family over \emph{all} matchings.

The genuinely non-elementary inputs are isolated in Proposition~\ref{prop:imported-tools} and in the
approximate counting routine for bounded-degree graphs used in Lemma~\ref{lem:compute-alpha}.  Everything
else in the testing argument is elementary.
\end{remark}

\section{Lower Bound}
\label{sec:lower-bound}

\subsection{Proof setup}

We now introduce the graph family used to prove Theorem~\ref{thm:main-lower}.  For a graph $G$,
write $U_G$ for the uniform distribution on $\Mcal(G)$.

For integers $t,q\ge 1$, let
\[
H_q:=K_{2q,2q},
\qquad
G_{t,q}:=\bigsqcup_{r=1}^t H_q^{(r)},
\]
where the $H_q^{(r)}$ are disjoint copies of $H_q$.
Then
\[
|V(G_{t,q})|=4tq,
\qquad
|E(G_{t,q})|=4tq^2.
\]
The proof builds a finite family of checkerboard perturbations on these components and bounds the
second moment of the corresponding mixture over random signs.

\begin{proposition}[Lower bound on the construction family]
\label{prop:main-Gtq}
There are absolute constants $c_0>0$ and $\varepsilon_0>0$ such that for all integers $t,q\ge 1$
and all $0<\varepsilon\le \varepsilon_0$, any tester that, given i.i.d.\ samples from an unknown
Gibbs distribution over matchings of $G_{t,q}$, distinguishes
\[
H_0:\ D_w=U_{G_{t,q}}
\qquad\text{from}\qquad
H_1:\ \TV(D_w,U_{G_{t,q}})\ge \varepsilon
\]
with success probability at least $2/3$ must use at least
\[
c_0\frac{q\sqrt t}{\varepsilon^2}
=c_0\frac{\sqrt{|E(G_{t,q})|}}{2\varepsilon^2}
\]
samples.
\end{proposition}

\subsection{The checkerboard family on \texorpdfstring{$K_{2q,2q}$}{K2q,2q}}

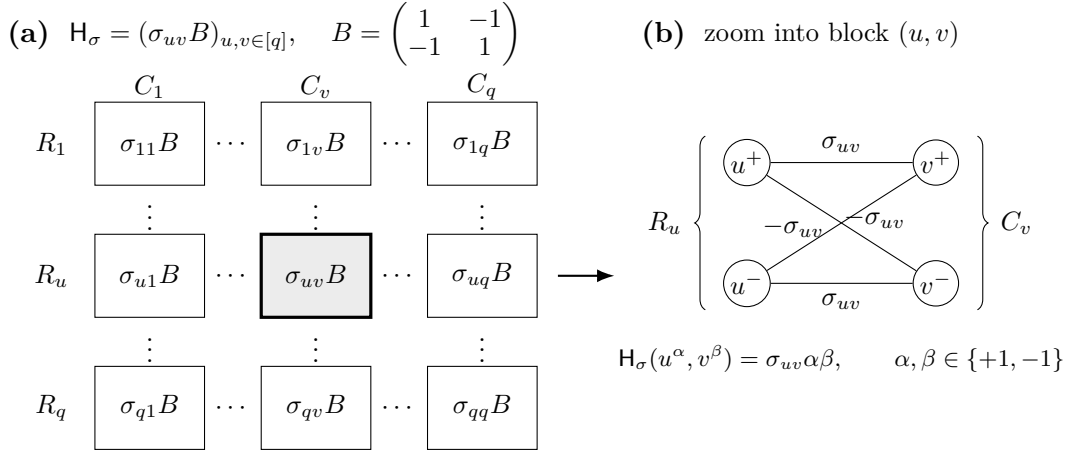
\begin{figure}[t!]
\centering
\makebox[\textwidth][c]{%
\begin{tikzpicture}[
    >=Latex,
    every node/.style={font=\small},
    block/.style={draw, minimum width=1.45cm, minimum height=1.1cm, align=center},
    hi/.style={block, fill=gray!15, very thick},
    vtx/.style={circle, draw, inner sep=1pt, minimum size=4.8mm},
    expl/.style={font=\footnotesize, align=center}
]

\node[font=\bfseries] at (0.2,1.9) {(a)};
\node[anchor=west] at (0.6,1.9)
{$\mathsf{H}_\sigma=(\sigma_{uv}B)_{u,v\in[q]}$,
\quad $B=\begin{pmatrix}1&-1\\[-0.5mm]-1&1\end{pmatrix}$};

\node at (1.8,1.2) {$C_1$};
\node at (4.0,1.2) {$C_v$};
\node at (6.2,1.2) {$C_q$};

\node at (0.5,0.45) {$R_1$};
\node at (0.5,-1.30) {$R_u$};
\node at (0.5,-3.05) {$R_q$};

\node[block] (a11) at (1.8,0.45) {$\sigma_{11}B$};
\node[block] (a1v) at (4.0,0.45) {$\sigma_{1v}B$};
\node[block] (a1q) at (6.2,0.45) {$\sigma_{1q}B$};

\node[block] (au1) at (1.8,-1.30) {$\sigma_{u1}B$};
\node[hi]    (auv) at (4.0,-1.30) {$\sigma_{uv}B$};
\node[block] (auq) at (6.2,-1.30) {$\sigma_{uq}B$};

\node[block] (aq1) at (1.8,-3.05) {$\sigma_{q1}B$};
\node[block] (aqv) at (4.0,-3.05) {$\sigma_{qv}B$};
\node[block] (aqq) at (6.2,-3.05) {$\sigma_{qq}B$};

\node at (2.9,0.45) {$\cdots$};
\node at (5.1,0.45) {$\cdots$};

\node at (2.9,-1.30) {$\cdots$};
\node at (5.1,-1.30) {$\cdots$};

\node at (2.9,-3.05) {$\cdots$};
\node at (5.1,-3.05) {$\cdots$};

\node at (1.8,-0.425) {$\vdots$};
\node at (4.0,-0.425) {$\vdots$};
\node at (6.2,-0.425) {$\vdots$};

\node at (1.8,-2.175) {$\vdots$};
\node at (4.0,-2.175) {$\vdots$};
\node at (6.2,-2.175) {$\vdots$};

\draw[->, thick] (7.20,-1.30) -- (7.95,-1.30);

\node[font=\bfseries] at (8.6,1.9) {(b)};
\node[anchor=west] at (9.0,1.9) {zoom into block $(u,v)$};

\node[vtx] (up) at (9.7,0.2) {$u^+$};
\node[vtx] (um) at (9.7,-1.4) {$u^-$};

\node[vtx] (vp) at (12.2,0.2) {$v^+$};
\node[vtx] (vm) at (12.2,-1.4) {$v^-$};

\draw (up) -- node[above] {$\sigma_{uv}$} (vp);
\draw (up) -- node[pos=.45, right] {$-\sigma_{uv}$} (vm);
\draw (um) -- node[pos=.45, left] {$-\sigma_{uv}$} (vp);
\draw (um) -- node[below] {$\sigma_{uv}$} (vm);

\draw[decorate, decoration={brace, amplitude=4pt, mirror}]
    (9.15,0.55) -- (9.15,-1.75)
    node[midway, left=5pt] {$R_u$};

\draw[decorate, decoration={brace, amplitude=4pt}]
    (12.75,0.55) -- (12.75,-1.75)
    node[midway, right=5pt] {$C_v$};

\node[expl] at (10.95,-2.4)
{$\mathsf{H}_\sigma(u^\alpha,v^\beta)=\sigma_{uv}\alpha\beta,\qquad \alpha,\beta\in\{+1,-1\}$};

\end{tikzpicture}
}%
\caption{Checkerboard perturbation on $H_q=K_{2q,2q}$.  
The left vertices are partitioned into row-pairs $R_u=\{u^+,u^-\}$ and the right vertices into column-pairs $C_v=\{v^+,v^-\}$.  
Each $(u,v)$-block carries the signed checkerboard pattern $\sigma_{uv}B$, where
$B=\begin{pmatrix}1&-1\\[-0.5mm]-1&1\end{pmatrix}$.  
In particular, every row sum and every column sum of $\mathsf{H}_\sigma$ is zero.}
\label{fig:checkerboard-family}
\end{figure}

Fix $q \ge 1$ and write $H_q = K_{2q,2q}$. Partition the left vertices into row-pairs
\[
R_u := \{u^+,u^-\}, \qquad u \in [q],
\]
and the right vertices into column-pairs
\[
C_v := \{v^+,v^-\}, \qquad v \in [q].
\]
Let
\[
B := \begin{pmatrix}1 & -1 \\ -1 & 1\end{pmatrix}.
\]
For each sign matrix $\sigma = (\sigma_{uv})_{u,v\in[q]} \in \{\pm 1\}^{q\times q}$, define the $2q\times 2q$ block matrix
\[
\mathsf{H}_\sigma := (\sigma_{uv} B)_{u,v\in[q]}.
\]
Equivalently,
\[
\mathsf{H}_\sigma(u^\alpha,v^\beta) = \sigma_{uv}\,\alpha\beta,
\qquad
\alpha,\beta \in \{+1,-1\}.
\]
This checkerboard construction is illustrated in Figure~\ref{fig:checkerboard-family}.
For any matrix $H \in \R^{2q\times 2q}$ and matching $M$ of $H_q$, set
\[
T_H(M) := \sum_{e\in M} H_e.
\]
We abbreviate
\[
T_\sigma := T_{\mathsf{H}_\sigma}.
\]
Let $U_q$ denote the uniform distribution on all matchings of $H_q$.

We begin with a few elementary facts.

\begin{lemma}\label{lem:block-basic}
For every $\sigma,\tau \in \{\pm 1\}^{q\times q}$:
\begin{enumerate}[label=(\roman*)]
\item every row sum and every column sum of $\mathsf{H}_\sigma$ is $0$;
\item
\[
\ip{\mathsf{H}_\sigma}{\mathsf{H}_\tau} = 4 \sum_{u,v=1}^q \sigma_{uv}\tau_{uv};
\]
\item in particular,
\[
\norm{\mathsf{H}_\sigma}_\F^2 = 4q^2.
\]
\end{enumerate}
Moreover, for every block matrix
\[
\mathsf{H} = \sum_{u,v=1}^q a_{uv} B_{uv}
\]
with $B_{uv}$ equal to $B$ in block $(u,v)$ and $0$ elsewhere, the distribution of $T_{\mathsf{H}}$ under $U_q$ is symmetric about $0$.
\end{lemma}

\begin{proof}
The first two claims are immediate from the block form. For the symmetry statement, swap the two vertices inside each column-pair $C_v$. This is an automorphism of $H_q$, it preserves the uniform distribution on matchings, and it flips the sign of every checkerboard block. Hence $T_{\mathsf{H}}$ is sent to $-T_{\mathsf{H}}$ while the distribution is unchanged.
\end{proof}

The next lemma gives an exact covariance formula on a complete bipartite graph.

\begin{lemma}[Isotropic covariance on $K_{a,b}$]\label{lem:isotropic-Kab}
Let $a,b\ge 2$ and $0\le k\le \min\{a,b\}$. Let $G = K_{a,b}$, let $U_k$ be the uniform distribution on the set of $k$-matchings of $G$, and for $H \in \R^{a\times b}$ define
\[
T_H(M) := \sum_{e\in M} H_e.
\]
Assume that $H\one = 0$ and $\one^\top H = 0$, and likewise for $H'$. Then
\[
\E_{U_k}[T_H] = 0,
\]
and
\[
\E_{U_k}[T_H T_{H'}]
=
\left(
\frac{k}{ab}
+
\frac{k(k-1)}{ab(a-1)(b-1)}
\right)
\ip{H}{H'}.
\]
\end{lemma}

\begin{proof}
By edge-transitivity,
\[
\E_{U_k}[T_H]
=
\sum_{i,j} H_{ij} \Pbb_{U_k}[(i,j)\in M]
=
\frac{k}{ab} \sum_{i,j} H_{ij}
=
0.
\]
For the second moment, let
\[
p_1 := \Pbb_{U_k}[(i,j)\in M] = \frac{k}{ab},
\qquad
p_2 := \Pbb_{U_k}[(i,j),(i',j')\in M] = \frac{k(k-1)}{ab(a-1)(b-1)}
\]
for disjoint edges $(i,j)$ and $(i',j')$. Then
\[
\E_{U_k}[T_H T_{H'}]
=
p_1 \sum_{i,j} H_{ij}H'_{ij}
+
p_2 \sum_{\substack{i\ne i' \\ j\ne j'}} H_{ij} H'_{i'j'}.
\]
It remains to show
\[
\sum_{\substack{i\ne i' \\ j\ne j'}} H_{ij} H'_{i'j'} = \ip{H}{H'}.
\]
Let $A := \ip{H}{H'}$. Since all row sums and all column sums vanish,
\begin{align*}
0
&= \sum_{i,i',j,j'} H_{ij}H'_{i'j'} \\
&= A
+ \sum_i \sum_{j\ne j'} H_{ij}H'_{ij'}
+ \sum_j \sum_{i\ne i'} H_{ij}H'_{i'j}
+ \sum_{\substack{i\ne i' \\ j\ne j'}} H_{ij}H'_{i'j'} \\
&= A - A - A + \sum_{\substack{i\ne i' \\ j\ne j'}} H_{ij}H'_{i'j'},
\end{align*}
which proves the identity.
\end{proof}

For the unconditional uniform distribution on all matchings of $H_q$, let $K_q := |M|$ under $U_q$ and set
\[
\beta_q := \E\left[ \frac{K_q}{4q^2} + \frac{K_q(K_q-1)}{4q^2(2q-1)^2} \right].
\]
Then Lemma~\ref{lem:isotropic-Kab} and Lemma~\ref{lem:block-basic} imply
\[
\Cov_{U_q}(T_\sigma,T_\tau) = 4\beta_q \sum_{u,v=1}^q \sigma_{uv}\tau_{uv}.
\]
In particular,
\[
\Var_{U_q}(T_\sigma) = 4\beta_q q^2.
\]

\begin{lemma}\label{lem:beta-bounds}
For every $q \ge 1$,
\[
\frac{1}{8q} \le \beta_q \le \frac{1}{q}.
\]
Consequently,
\[
\frac{q}{2} \le \Var_{U_q}(T_\sigma) \le 4q
\qquad\text{for every } \sigma \in \{\pm1\}^{q\times q}.
\]
\end{lemma}

\begin{proof}
For $q=1$, direct computation gives $\beta_1=3/7$, so the claimed bounds hold.
Assume henceforth that $q\ge 2$.

The number of $k$-matchings in $H_q = K_{2q,2q}$ is
\[
N_k = \binom{2q}{k}^2 k!,
\qquad 0 \le k \le 2q.
\]
Moreover,
\[
\frac{N_{k+1}}{N_k} = \frac{(2q-k)^2}{k+1}.
\]
For $0 \le k \le q-1$ this ratio is at least
\[
\frac{(q+1)^2}{q} > q.
\]
Thus $N_j \le q^{-(q-1-j)}N_{q-1}$ for $0\le j\le q-1$, so
\[
\sum_{k=0}^{q-1}N_k
\le
N_{q-1}\sum_{m=0}^{q-1}q^{-m}
< 2N_{q-1}
\le qN_{q-1}.
\]
Hence $N_q > q N_{q-1} > \sum_{k=0}^{q-1} N_k$, and therefore
\[
\Pbb(K_q \ge q) > \frac12.
\]
For the lower bound,
\[
\beta_q \ge \E\left[ \frac{K_q}{4q^2} \one_{\{K_q\ge q\}} \right] \ge \frac{1}{8q}.
\]
For the upper bound, $K_q \le 2q$, so
\[
\beta_q
\le
\frac{\E[K_q]}{4q^2} + \frac{\E[K_q(K_q-1)]}{4q^2(2q-1)^2}
\le
\frac{2q}{4q^2} + \frac{2q(2q-1)}{4q^2(2q-1)^2}
<
\frac1q.
\]
The variance bounds are immediate from $\Var(T_\sigma) = 4\beta_q q^2$.
\end{proof}

\subsection{A quadratic log-MGF expansion on \texorpdfstring{$K_{2q,2q}$}{K2q,2q}}

We prove a local quadratic expansion for the logarithm of the moment-generating function for
checkerboard scores on $H_q$.

\subsubsection{Row-pair martingales and an exact \texorpdfstring{$\log\cosh$}{log cosh} representation}

Fix a block matrix
\[
\mathsf{H} = \sum_{u,v=1}^q a_{uv} B_{uv}
\qquad\text{with}\qquad |a_{uv}| \le 2.
\]
Expose a matching of $H_q$ row-pair by row-pair. Let $\mathcal{F}_u$ be the sigma-field generated by the restriction of the matching to $R_1,\dots,R_u$, and let $\mathcal{F}_0$ be the trivial sigma-field. For $u \in [q]$, let
\[
Y_u(M) := \sum_{\substack{e\in M \\ e\text{ incident to }R_u}} \mathsf{H}_e.
\]
Then
\[
T_{\mathsf{H}} = \sum_{u=1}^q Y_u.
\]

\begin{lemma}\label{lem:rowpair-martingale}
The sequence $(Y_u,\mathcal{F}_u)_{u=1}^q$ is a martingale-difference sequence under $U_q$:
\[
\E[Y_u \mid \mathcal{F}_{u-1}] = 0.
\]
Moreover,
\[
|Y_u| \le 2 \max_v |a_{uv}| \le 4
\qquad\text{almost surely.}
\]
\end{lemma}

\begin{proof}
Conditional on $\mathcal{F}_{u-1}$, the two vertices $u^+$ and $u^-$ are symmetric in the
remaining complete bipartite graph on the unrevealed vertices. Swapping them preserves the
conditional distribution and flips the sign of $Y_u$, because every block on row-pair $R_u$ is of
checkerboard form. Hence the conditional distribution of $Y_u$ is symmetric and therefore has mean
$0$.

Since $R_u$ contains only two left vertices and each matched edge contributes at most $\max_v |a_{uv}|$ in absolute value, the deterministic bound follows.
\end{proof}

Let
\[
\ell(x) := \log\cosh x.
\]

\begin{lemma}[Log-cosh identities from row-pair symmetries]\label{lem:exact-cosh}
For every $s\in\R$,
\[
\E_{U_q}\exp\!\left(
 s T_{\mathsf{H}} - \sum_{u=1}^q \ell(sY_u)
\right)
= 1
\]
and
\[
\E_{U_q} e^{s T_{\mathsf{H}}}
=
\E_{U_q} \exp\!\left(\sum_{u=1}^q \ell(sY_u)\right).
\]
\end{lemma}

\begin{proof}
For each $u\in[q]$, let $\tau_u$ be the graph automorphism of $H_q$ that swaps the two vertices
$u^+$ and $u^-$ in the row-pair $R_u$. For
\[
\varepsilon=(\varepsilon_1,\dots,\varepsilon_q)\in\{\pm1\}^q,
\]
let $\tau_\varepsilon$ be the composition of the maps $\tau_u$ with $\varepsilon_u=-1$.
Since $U_q$ is the uniform distribution on all matchings of $H_q$, it is invariant under every
$\tau_\varepsilon$.

Fix a matching $M$. Swapping the row-pair $R_u$ flips the sign of every checkerboard entry on the
edges incident to $R_u$, while swapping a different row-pair does not affect those edges. Hence
\[
Y_u(\tau_\varepsilon M)=\varepsilon_u Y_u(M)
\qquad (u\in[q]).
\]
Therefore
\[
T_{\mathsf{H}}(\tau_\varepsilon M)=\sum_{u=1}^q \varepsilon_u Y_u(M),
\qquad
\sum_{u=1}^q \ell\bigl(sY_u(\tau_\varepsilon M)\bigr)=\sum_{u=1}^q \ell\bigl(sY_u(M)\bigr),
\]
since $\ell$ is even.

Averaging over all sign patterns gives, for each fixed $M$,
\begin{align*}
2^{-q}\sum_{\varepsilon\in\{\pm1\}^q} e^{sT_{\mathsf{H}}(\tau_\varepsilon M)}
&=
2^{-q}\sum_{\varepsilon\in\{\pm1\}^q}
\exp\!\left(\sum_{u=1}^q s\varepsilon_u Y_u(M)\right) \\
&=
\prod_{u=1}^q \frac{e^{sY_u(M)}+e^{-sY_u(M)}}{2}
=
\exp\!\left(\sum_{u=1}^q \ell\bigl(sY_u(M)\bigr)\right),
\end{align*}
and likewise
\begin{align*}
2^{-q}\sum_{\varepsilon\in\{\pm1\}^q}
\exp\!\left(
 sT_{\mathsf{H}}(\tau_\varepsilon M)-\sum_{u=1}^q \ell\bigl(sY_u(\tau_\varepsilon M)\bigr)
\right)
&=
\prod_{u=1}^q
\left(
 e^{-\ell(sY_u(M))}\frac{e^{sY_u(M)}+e^{-sY_u(M)}}{2}
\right)
=1.
\end{align*}
Taking expectation under $U_q$ and using invariance under the automorphisms $\tau_\varepsilon$
proves both identities.
\end{proof}

\begin{corollary}\label{cor:subgaussian-one-copy}
There exists an absolute constant $C$ such that for every block matrix $\mathsf{H}$ with $|a_{uv}|\le 2$ and every $s \in \R$,
\[
\E_{U_q} e^{s T_{\mathsf{H}}} \le e^{C q s^2}.
\]
Consequently,
\[
\Pbb_{U_q}(|T_{\mathsf{H}}| \ge x) \le 2 e^{-x^2/(Cq)}
\qquad\text{for every }x\ge 0,
\]
and
\[
\E_{U_q}[T_{\mathsf{H}}^4] \le C q^2.
\]
\end{corollary}

\begin{proof}
Let
\[
Q := \sum_{u=1}^q Y_u^2 .
\]
By Lemma~\ref{lem:exact-cosh} and the pointwise bound $\ell(x)\le x^2/2$,
\[
\E e^{sT_{\mathsf{H}}}
=
\E \exp\!\left(\sum_{u=1}^q \ell(sY_u)\right)
\le
\E \exp\!\left(\frac{s^2}{2}\sum_{u=1}^q Y_u^2\right)
=
\E e^{\frac{s^2}{2}Q}.
\]
Since $|Y_u|\le 4$, we have $Q\le 16q$ almost surely, hence
\[
\E e^{sT_{\mathsf{H}}} \le e^{8 q s^2}.
\]
The MGF bound follows. The tail estimate is a standard Chernoff bound, and the fourth-moment estimate follows by integrating the tail.
\end{proof}

\subsubsection{Occupation probabilities and a smoothness estimate}

For integers $1 \le a \le b$, let $U_{a,b}$ be the uniform distribution on all matchings of $K_{a,b}$, and let $K_{a,b}$ denote the matching size under $U_{a,b}$. Define
\[
\pi_1(a,b) := U_{a,b}(e\in M) = \frac{\E[K_{a,b}]}{ab}
\]
for any fixed edge $e$, and
\[
\pi_2(a,b) := U_{a,b}(e,f\in M) = \frac{\E[K_{a,b}(K_{a,b}-1)]}{a(a-1)b(b-1)}
\]
for any fixed ordered pair of disjoint edges $(e,f)$.

\begin{lemma}[Matching-size monotonicity]\label{lem:size-monotone}
Fix integers $1 \le a \le b$. Then the distribution of $K_{a,b}$ is stochastically increasing in
$b$. More precisely, $K_{a,b+1}$ stochastically dominates $K_{a,b}$.
\end{lemma}

\begin{proof}
The number of $k$-matchings in $K_{a,b}$ is
\[
N_k(a,b) := \binom{a}{k}\binom{b}{k}k!,
\qquad 0\le k\le a.
\]
Hence
\[
\frac{N_k(a,b+1)}{N_k(a,b)} = \frac{b+1}{b+1-k},
\]
which is increasing in $k$. Thus the probability mass function of $K_{a,b+1}$ is obtained from that of $K_{a,b}$ by reweighting with an increasing function of $k$. This is the monotone-likelihood-ratio order, and in particular it implies stochastic domination.
\end{proof}

\begin{lemma}\label{lem:occupation-smooth}
For every integers $2 \le a \le b$,
\[
0 \le \pi_1(a,b) \le \frac{1}{b},
\qquad
0 \le \pi_2(a,b) \le \frac{1}{b(b-1)}.
\]
Moreover,
\[
|\pi_1(a,b+1)-\pi_1(a,b)| \le \frac{2}{ab},
\qquad
|\pi_2(a,b+1)-\pi_2(a,b)| \le \frac{C}{ab^2}
\]
for an absolute constant $C$. Consequently, if $b'\ge a$ and $|b-b'|\le 2$, then
\[
|\pi_1(a,b')-\pi_1(a,b)| \le \frac{C}{ab},
\qquad
|\pi_2(a,b')-\pi_2(a,b)| \le \frac{C}{ab^2}.
\]
\end{lemma}

\begin{proof}
Since $K_{a,b} \le a$,
\[
\pi_1(a,b) = \frac{\E[K_{a,b}]}{ab} \le \frac{1}{b},
\qquad
\pi_2(a,b) = \frac{\E[K_{a,b}(K_{a,b}-1)]}{a(a-1)b(b-1)} \le \frac{1}{b(b-1)}.
\]

Fix a new right vertex $z$ and sample a matching of $K_{a,b+1}$ uniformly. Delete the edge incident to $z$ if such an edge exists, and let $L$ be the size of the resulting matching on $K_{a,b}$.

Every $k$-matching of $K_{a,b}$ has exactly $a-k+1$ extensions to a matching of $K_{a,b+1}$: either the new vertex $z$ is unused, or it is matched to one of the $a-k$ unmatched left vertices. Therefore
\[
\Pbb(L=k) \propto (a-k+1) N_k(a,b).
\]
Since the factor $a-k+1$ is decreasing in $k$, the distribution of $L$ is the distribution of
$K_{a,b}$ reweighted by a decreasing function of $k$. Hence for every increasing function $g$,
\[
\E[g(L)] \le \E[g(K_{a,b})].
\]
In particular,
\[
\E[L] \le \E[K_{a,b}],
\qquad
\E[L(L-1)] \le \E[K_{a,b}(K_{a,b}-1)].
\]

Let $K' := K_{a,b+1}$ and let $I_z$ be the indicator that $z$ is used. Since the right vertices are symmetric,
\[
\E[I_z] = \frac{\E[K']}{b+1}.
\]
Also $L=K' - I_z$, so
\[
\E[L] = \frac{b}{b+1}\E[K'].
\]
Combining this with $\E[L]\le \E[K_{a,b}]$ yields
\[
\E[K'] \le \frac{b+1}{b}\E[K_{a,b}],
\]
while Lemma~\ref{lem:size-monotone} gives the reverse inequality $\E[K'] \ge \E[K_{a,b}]$. Hence
\[
0 \le \E[K'] - \E[K_{a,b}] \le \frac{1}{b}\E[K_{a,b}] \le \frac{a}{b}.
\]
Therefore
\begin{align*}
|\pi_1(a,b+1)-\pi_1(a,b)|
&= \left| \frac{\E[K']}{a(b+1)} - \frac{\E[K_{a,b}]}{ab} \right| \\
&\le \frac{|\E[K'] - \E[K_{a,b}]|}{a(b+1)}
 + \frac{\E[K_{a,b}]}{a}\left|\frac{1}{b+1}-\frac{1}{b}\right| \\
&\le \frac{1}{b(b+1)} + \frac{1}{b(b+1)}
\le \frac{2}{ab},
\end{align*}
since $a\le b$.

Now let
\[
F_{a,b} := \E[K_{a,b}(K_{a,b}-1)],
\qquad
F' := F_{a,b+1}.
\]
Because
\[
L(L-1) = K'(K'-1) - 2 I_z (K'-1),
\]
and because symmetry among right vertices gives
\[
\E[I_z(K'-1)] = \frac{F'}{b+1},
\]
we obtain
\[
\E[L(L-1)] = \frac{b-1}{b+1}F'.
\]
Since $\E[L(L-1)]\le F_{a,b}$, we conclude that
\[
F' \le \frac{b+1}{b-1}F_{a,b}.
\]
On the other hand, Lemma~\ref{lem:size-monotone} applied to the increasing function $k\mapsto k(k-1)$ gives $F' \ge F_{a,b}$. Therefore
\[
0 \le F' - F_{a,b} \le \frac{2}{b-1}F_{a,b} \le \frac{2a(a-1)}{b-1}.
\]
Hence
\begin{align*}
|\pi_2(a,b+1)-\pi_2(a,b)|
&= \left|
\frac{F'}{a(a-1)b(b+1)}
-
\frac{F_{a,b}}{a(a-1)b(b-1)}
\right| \\
&\le
\frac{|F' - F_{a,b}|}{a(a-1)b(b+1)}
+
\frac{F_{a,b}}{a(a-1)}
\left|
\frac{1}{b(b+1)} - \frac{1}{b(b-1)}
\right| \\
&\le
\frac{2}{(b-1)b(b+1)}
+
\frac{2}{(b-1)b(b+1)}
\le
\frac{C}{ab^2},
\end{align*}
because $a\le b$.
The final bounds for $b'\ge a$ and $|b-b'|\le 2$ follow by applying the one-step estimate at most twice.
\end{proof}

\subsubsection{A reduced second-moment formula}

\begin{lemma}\label{lem:reduced-second-moment}
Let $2 \le a \le b$, let $U_{a,b}$ be the uniform distribution on all matchings of $K_{a,b}$, and let $G \in \R^{a\times b}$. Define
\[
T_G(M) := \sum_{e\in M} G_e.
\]
Then
\[
\E_{U_{a,b}}[T_G^2]
=
\pi_1(a,b) \norm{G}_\F^2
+
\pi_2(a,b) \sum_{\substack{i\ne i' \\ j\ne j'}} G_{ij}G_{i'j'}.
\]
If, in addition, every column sum of $G$ vanishes, then
\[
\E_{U_{a,b}}[T_G^2]
=
\bigl(\pi_1(a,b)+\pi_2(a,b)\bigr) \norm{G}_\F^2
-
\pi_2(a,b) \norm{r(G)}_2^2,
\]
where $r(G) \in \R^a$ is the vector of row sums.
\end{lemma}

\begin{proof}
Expanding $T_G^2$ and using that two adjacent edges cannot both belong to a matching gives
\[
\E[T_G^2]
=
\pi_1(a,b) \sum_{i,j} G_{ij}^2
+
\pi_2(a,b) \sum_{\substack{i\ne i' \\ j\ne j'}} G_{ij}G_{i'j'}.
\]
This proves the first identity.

For the second, let $S := \sum_{i,j} G_{ij}$ and let $c(G) \in \R^b$ be the vector of column sums. The elementary identity
\[
\sum_{\substack{i\ne i' \\ j\ne j'}} G_{ij}G_{i'j'}
=
S^2 - \norm{r(G)}_2^2 - \norm{c(G)}_2^2 + \norm{G}_\F^2
\]
shows that if every column sum vanishes, then $c(G)=0$ and also $S=0$, hence
\[
\sum_{\substack{i\ne i' \\ j\ne j'}} G_{ij}G_{i'j'} = \norm{G}_\F^2 - \norm{r(G)}_2^2.
\]
Substituting this into the first formula proves the claim.
\end{proof}

\subsubsection{Concentration of the realized quadratic variation}

For the fixed block matrix $\mathsf{H}$, keep the row-pair martingale differences $Y_u$ from Lemma~\ref{lem:rowpair-martingale}, and define
\[
Q := \sum_{u=1}^q Y_u^2.
\]
For $s = 0,1,\dots,q$, let
\[
S_s := \sum_{u>s} Y_u,
\qquad
R_s := \E[S_s^2 \mid \mathcal{F}_s].
\]

\begin{lemma}\label{lem:future-second-stable}
There is an absolute constant $C$ with the following property.

For each $s=1,\dots,q-1$, conditional on $\mathcal{F}_{s-1}$, the random variable $R_s$ ranges over an interval of length at most $C$.
Hence, if
\[
M_s := \E[Q \mid \mathcal{F}_s],
\qquad 0\le s\le q,
\]
then
\[
|M_s - M_{s-1}| \le C
\qquad\text{almost surely for every }1\le s\le q .
\]
\end{lemma}

\begin{proof}
Fix $s \in \{1,\dots,q-1\}$. Conditional on $\mathcal{F}_s$, the future matching lives on a complete bipartite graph $K_{a_s,b_s}$, where
\[
a_s := 2(q-s)
\]
is the number of unrevealed left vertices and $b_s$ is the number of right vertices not yet used. Since each revealed row-pair can use at most two right vertices,
\[
b_s \ge 2q - 2s = a_s.
\]
Let $G_s$ be the restriction of $\mathsf{H}$ to the unrevealed left vertices and currently unused right vertices. Then $S_s = T_{G_s}$.

Every column sum of $G_s$ is zero: within each unrevealed row-pair the two entries in any fixed remaining column are opposite. Therefore Lemma~\ref{lem:reduced-second-moment} yields
\[
R_s
=
\bigl(\pi_1(a_s,b_s)+\pi_2(a_s,b_s)\bigr) \norm{G_s}_\F^2
-
\pi_2(a_s,b_s) \norm{r(G_s)}_2^2 .
\]

Now fix $\mathcal{F}_{s-1}$. Different outcomes of the $s$th row-pair remove from the current right side either $0$, $1$, or $2$ right vertices. Hence any two possible resulting right sets differ in at most four vertices. Therefore the corresponding matrices $G_s$ and $G_s'$ differ by adding or deleting at most four columns.

Because each entry has absolute value at most $2$, adding or deleting one column changes the Frobenius norm squared by at most $4a_s$, and adding or deleting at most four columns changes it by at most
\[
|\norm{G_s}_\F^2 - \norm{G_s'}_\F^2| \le 16 a_s.
\]
Likewise, each row sum changes by at most $8$ when at most four columns are altered. Since every row sum has magnitude at most $2b_s$,
\[
|\norm{r(G_s)}_2^2 - \norm{r(G_s')}_2^2|
\le
\sum_{i=1}^{a_s} \bigl(2|r_i|\,|\Delta_i| + \Delta_i^2 \bigr)
\le
C a_s b_s.
\]
Also,
\[
\norm{G_s}_\F^2 \le 4 a_s b_s,
\qquad
\norm{r(G_s)}_2^2 \le 4 a_s b_s^2,
\]
and the same bounds hold for $G_s'$.

Since $|b_s-b_s'|\le 2$, Lemma~\ref{lem:occupation-smooth} gives
\[
|\pi_1(a_s,b_s')-\pi_1(a_s,b_s)| \le \frac{C}{a_s b_s},
\qquad
|\pi_2(a_s,b_s')-\pi_2(a_s,b_s)| \le \frac{C}{a_s b_s^2}.
\]
Combining these estimates with the formula for $R_s$ shows that every contribution to $|R_s-R_s'|$ is bounded by an absolute constant. Therefore the conditional range of $R_s$ is $O(1)$.

Next, for every $s$,
\[
M_s
=
\sum_{u\le s} Y_u^2 + \E\!\left[\sum_{u>s} Y_u^2 \,\middle|\, \mathcal{F}_s\right].
\]
Because $(Y_u,\mathcal{F}_u)$ is a martingale-difference sequence, the future conditional cross-terms vanish:
for $u<v$ with $u,v>s$,
\[
\E[Y_uY_v \mid \mathcal{F}_s]
=
\E\!\left[\E[Y_uY_v \mid \mathcal{F}_{v-1}] \,\middle|\, \mathcal{F}_s\right]
=
\E\!\left[Y_u \E[Y_v\mid \mathcal{F}_{v-1}] \,\middle|\, \mathcal{F}_s\right]
=
0.
\]
Hence
\[
\E\!\left[\sum_{u>s} Y_u^2 \,\middle|\, \mathcal{F}_s\right]
=
\E[S_s^2 \mid \mathcal{F}_s]
=
R_s,
\]
so
\[
M_s = \sum_{u\le s} Y_u^2 + R_s.
\]

For $1\le s\le q-1$ we therefore have
\[
M_s - M_{s-1}
=
\bigl(Y_s^2 + R_s\bigr) - \E\!\left[Y_s^2 + R_s \,\middle|\, \mathcal{F}_{s-1}\right].
\]
Since $0\le Y_s^2\le 16$ and the conditional range of $R_s$ is bounded by an absolute constant, the conditional range of $Y_s^2+R_s$ is also bounded by an absolute constant, and thus
\[
|M_s-M_{s-1}| \le C.
\]
For $s=q$, one has $R_q=0$ and therefore
\[
M_q-M_{q-1}=Y_q^2-\E[Y_q^2\mid \mathcal{F}_{q-1}],
\]
which is bounded by $16$. After enlarging $C$ if necessary, the claimed increment bound holds for all $1\le s\le q$.
\end{proof}

\begin{corollary}\label{cor:Q-concentration}
There exists an absolute constant $C$ such that for every $\lambda \in \R$,
\[
\log \E e^{\lambda (Q-\E Q)} \le C q \lambda^2.
\]
In particular,
\[
\Var(Q) \le C q.
\]
Moreover,
\[
\E Q = \Var(T_{\mathsf{H}}).
\]
\end{corollary}

\begin{proof}
Apply Azuma--Hoeffding to the Doob martingale $M_s = \E[Q\mid \mathcal{F}_s]$ using Lemma~\ref{lem:future-second-stable}. The identity $\E Q = \Var(T_{\mathsf{H}})$ follows from the orthogonality of martingale differences:
\[
\Var\!\left(\sum_{u=1}^q Y_u\right)=\sum_{u=1}^q \E[Y_u^2]=\E Q.
\]
\end{proof}

The next lemma gives the local quadratic expansion used in the second-moment argument.

\begin{lemma}[Quadratic log-MGF expansion]\label{lem:key-logmgf}
There are absolute constants $c_*,C_* > 0$ such that for every block matrix
\[
\mathsf{H} = \sum_{u,v=1}^q a_{uv} B_{uv}
\qquad\text{with}\qquad |a_{uv}| \le 2,
\]
and every $|s| \le c_*$,
\[
\left|
\log \E_{U_q} e^{s T_{\mathsf{H}}}
-
\frac{s^2}{2} \Var_{U_q}(T_{\mathsf{H}})
\right|
\le C_* q s^4.
\]
\end{lemma}

\begin{proof}
By Lemma~\ref{lem:rowpair-martingale}, $|Y_u|\le 4$. Choose $c_*>0$ so small that $|sY_u|\le 1$ whenever $|s|\le c_*$. For such $s$, the Taylor expansion of $\ell(x)=\log\cosh x$ gives
\[
0 \le \frac{x^2}{2} - \ell(x) \le C x^4
\qquad (|x|\le 1).
\]
Hence, for each $u$,
\[
\frac{s^2}{2}Y_u^2 - C s^4 Y_u^4
\le
\ell(sY_u)
\le
\frac{s^2}{2}Y_u^2.
\]
Since $|Y_u|\le 4$, we have $Y_u^4 \le 16 Y_u^2$, so after summing over $u$ we obtain
\[
\left(\frac{s^2}{2} - C s^4\right)Q
\le
\sum_{u=1}^q \ell(sY_u)
\le
\frac{s^2}{2}Q .
\]
Lemma~\ref{lem:exact-cosh} therefore gives
\[
\E e^{(\frac{s^2}{2} - C s^4)Q}
\le
\E e^{sT_{\mathsf{H}}}
\le
\E e^{\frac{s^2}{2}Q}.
\]

For the upper bound, Corollary~\ref{cor:Q-concentration} yields
\[
\log \E e^{\frac{s^2}{2}Q}
=
\frac{s^2}{2}\E Q + \log \E e^{\frac{s^2}{2}(Q-\E Q)}
\le
\frac{s^2}{2}\E Q + C q s^4.
\]
For the lower bound, Jensen's inequality gives
\[
\log \E e^{(\frac{s^2}{2} - C s^4)Q}
\ge
\left(\frac{s^2}{2} - C s^4\right)\E Q.
\]
Finally, Corollary~\ref{cor:Q-concentration} also gives $\E Q = \Var(T_{\mathsf{H}})$, and the deterministic bound $Q\le 16q$ implies
\[
\E Q \le 16 q.
\]
Hence
\[
\frac{s^2}{2}\Var(T_{\mathsf{H}}) - C q s^4
\le
\log \E e^{sT_{\mathsf{H}}}
\le
\frac{s^2}{2}\Var(T_{\mathsf{H}}) + C q s^4.
\]
This completes the proof.
\end{proof}

\subsection{The lower bound on \texorpdfstring{$G_{t,q}$}{Gt,q}}

Fix integers $t,q\ge 1$. Let
\[
G_{t,q} = \bigsqcup_{r=1}^t H_q^{(r)}
\]
be the disjoint union of $t$ copies of $H_q$. For a sign collection
\[
\boldsymbol{\sigma} = \bigl(\sigma^{(1)},\dots,\sigma^{(t)}\bigr)
\in \bigl(\{\pm1\}^{q\times q}\bigr)^t,
\]
define the total score
\[
T_{\boldsymbol{\sigma}} := \sum_{r=1}^t T_{\sigma^{(r)}}^{(r)},
\]
where $T_{\sigma^{(r)}}^{(r)}$ is the checkerboard score on the $r$th copy. Let $U_{t,q}$ denote the uniform distribution on all matchings of $G_{t,q}$. Since the components are disjoint, $U_{t,q}$ is the product of $t$ independent copies of $U_q$.

\begin{lemma}\label{lem:total-score-moments}
There exist absolute constants $c,C>0$ such that for every $\boldsymbol{\sigma}$,
\[
 c t q \le \Var_{U_{t,q}}(T_{\boldsymbol{\sigma}}) \le C t q,
\qquad
\E_{U_{t,q}}[T_{\boldsymbol{\sigma}}^4] \le C (tq)^2 .
\]
Moreover, the distribution of $T_{\boldsymbol{\sigma}}$ under $U_{t,q}$ is symmetric about $0$.
\end{lemma}

\begin{proof}
The component scores are independent and centered. By Lemma~\ref{lem:beta-bounds}, each component has variance between $q/2$ and $4q$, so
\[
\Var(T_{\boldsymbol{\sigma}}) = \sum_{r=1}^t \Var(T_{\sigma^{(r)}}^{(r)}) \asymp tq.
\]
For the fourth moment, independence and centering give
\[
\E[T_{\boldsymbol{\sigma}}^4]
=
\sum_{r=1}^t \E[(T_{\sigma^{(r)}}^{(r)})^4]
+
6\sum_{1\le r < s \le t}
\E[(T_{\sigma^{(r)}}^{(r)})^2]\E[(T_{\sigma^{(s)}}^{(s)})^2].
\]
Using Corollary~\ref{cor:subgaussian-one-copy} and Lemma~\ref{lem:beta-bounds}, the first sum is $O(tq^2)$ and the second is $O(t^2 q^2)$, hence the total is $O((tq)^2)$.

Symmetry follows from Lemma~\ref{lem:block-basic}, because each component score is symmetric and the components are independent.
\end{proof}

For a parameter $\gamma > 0$ and a sign collection $\boldsymbol{\sigma}$, define edge weights on $G_{t,q}$ by
\[
w_e^{(\boldsymbol{\sigma})} := \exp\bigl(\gamma \mathsf{H}_{\sigma^{(r)}}(e)\bigr)
\]
for $e$ in the $r$th copy. Let $D_{\boldsymbol{\sigma}}$ be the corresponding Gibbs distribution over
matchings. Its likelihood ratio with respect to $U_{t,q}$ is
\[
L_{\boldsymbol{\sigma}}(M)
:=
\frac{d D_{\boldsymbol{\sigma}}}{dU_{t,q}}(M)
=
\exp\bigl(\gamma T_{\boldsymbol{\sigma}}(M) - \Psi_{\boldsymbol{\sigma}}(\gamma)\bigr),
\]
where
\[
\Psi_{\boldsymbol{\sigma}}(\gamma) := \log \E_{U_{t,q}} e^{\gamma T_{\boldsymbol{\sigma}}}.
\]

We first show that each constructed distribution is genuinely $\Omega(\varepsilon)$-far from $U_{t,q}$.

\begin{proposition}\label{prop:alt-far}
There are absolute constants $a,c_{\mathrm{far}},\varepsilon_0 > 0$ such that if
\[
\gamma = a \frac{\varepsilon}{\sqrt{tq}},
\qquad 0 < \varepsilon \le \varepsilon_0,
\]
then, for every sign collection $\boldsymbol{\sigma}$,
\[
\TV(D_{\boldsymbol{\sigma}},U_{t,q}) \ge c_{\mathrm{far}} \varepsilon .
\]
\end{proposition}

\begin{proof}
Fix $\boldsymbol{\sigma}$ and abbreviate $T := T_{\boldsymbol{\sigma}}$. By Lemma~\ref{lem:total-score-moments},
\[
\E[T^2] \asymp tq,
\qquad
\E[T^4] \le C(tq)^2 .
\]
Applying the Paley--Zygmund inequality to $T^2$ yields absolute constants $p_0,\eta>0$ such that
\[
U_{t,q}(T \ge \eta \sqrt{tq}) \ge p_0 .
\]
Indeed, $U_{t,q}(T^2 \ge \tfrac12 \E[T^2]) \ge c$, and symmetry of $T$ splits this two-sided event evenly between the positive and negative tails.

Choose $a$ and $\varepsilon_0$ so that $a\varepsilon_0\le c_*$, where $c_*$ is the constant in Lemma~\ref{lem:key-logmgf}; then $\gamma\le c_*$.
Since $T$ is a sum of $t$ independent component scores and each component satisfies Lemma~\ref{lem:key-logmgf}, we have
\[
\Psi_{\boldsymbol{\sigma}}(\gamma)
=
\frac{\gamma^2}{2}\Var(T) + O(tq\gamma^4)
=
O(\varepsilon^2) .
\]
On the event
\[
A := \{T \ge \eta \sqrt{tq}\}
\]
we therefore have
\[
L_{\boldsymbol{\sigma}}(M)
\ge
\exp\bigl(\gamma \eta \sqrt{tq} - \Psi_{\boldsymbol{\sigma}}(\gamma)\bigr)
=
\exp\bigl(a\eta \varepsilon - O(\varepsilon^2)\bigr).
\]
Therefore,
\begin{align*}
\TV(D_{\boldsymbol{\sigma}},U_{t,q})
&\ge D_{\boldsymbol{\sigma}}(A) - U_{t,q}(A) \\
&= \E_{U_{t,q}}[(L_{\boldsymbol{\sigma}}-1)\one_A] \\
&\ge U_{t,q}(A)\Bigl(e^{a\eta\varepsilon - O(\varepsilon^2)} - 1\Bigr).
\end{align*}
Choosing $a>0$ small and then $\varepsilon_0>0$ small enough, the last quantity is at least $c_{\mathrm{far}}\varepsilon$.
\end{proof}

Next we estimate the overlap of two constructed distributions under a random choice of signs.

\begin{lemma}\label{lem:pair-overlap}
Let $\boldsymbol{\sigma},\boldsymbol{\tau}$ be two sign collections, and define
\[
X(\boldsymbol{\sigma},\boldsymbol{\tau})
:=
\sum_{r=1}^t \sum_{u,v=1}^q \sigma^{(r)}_{uv} \tau^{(r)}_{uv}.
\]
If $|\gamma|\le c_*$, where $c_*$ is the constant in Lemma~\ref{lem:key-logmgf}, then
\[
\log \E_{U_{t,q}}[L_{\boldsymbol{\sigma}} L_{\boldsymbol{\tau}}]
=
4\beta_q \gamma^2 X(\boldsymbol{\sigma},\boldsymbol{\tau}) + O(tq\gamma^4),
\]
where $\beta_q$ is the coefficient defined before Lemma~\ref{lem:beta-bounds}.
\end{lemma}

\begin{proof}
Let $\Psi_{\boldsymbol{\sigma},\boldsymbol{\tau}}(\gamma)$ denote the log-MGF of $T_{\boldsymbol{\sigma}} + T_{\boldsymbol{\tau}}$ under $U_{t,q}$. Then
\[
\log \E[L_{\boldsymbol{\sigma}}L_{\boldsymbol{\tau}}]
=
\Psi_{\boldsymbol{\sigma},\boldsymbol{\tau}}(\gamma)
-
\Psi_{\boldsymbol{\sigma}}(\gamma)
-
\Psi_{\boldsymbol{\tau}}(\gamma).
\]
Each component of $T_{\boldsymbol{\sigma}} + T_{\boldsymbol{\tau}}$ is a checkerboard score with coefficients in $\{-2,0,2\}$, so Lemma~\ref{lem:key-logmgf} applies both to the single-score and to the summed-score log-MGFs. Therefore
\[
\log \E[L_{\boldsymbol{\sigma}}L_{\boldsymbol{\tau}}]
=
\frac{\gamma^2}{2}
\Bigl(
\Var(T_{\boldsymbol{\sigma}}+T_{\boldsymbol{\tau}})
-
\Var(T_{\boldsymbol{\sigma}})
-
\Var(T_{\boldsymbol{\tau}})
\Bigr)
+ O(tq\gamma^4).
\]
Thus
\[
\log \E[L_{\boldsymbol{\sigma}}L_{\boldsymbol{\tau}}]
=
\gamma^2 \Cov(T_{\boldsymbol{\sigma}},T_{\boldsymbol{\tau}}) + O(tq\gamma^4).
\]
Since the components are independent and each component covariance is given by Lemma~\ref{lem:isotropic-Kab} together with Lemma~\ref{lem:block-basic},
\[
\Cov(T_{\boldsymbol{\sigma}},T_{\boldsymbol{\tau}})
=
4\beta_q X(\boldsymbol{\sigma},\boldsymbol{\tau}).
\]
This completes the proof.
\end{proof}

\begin{proposition}\label{prop:mixture-chi2}
Let
\[
\overline{D}_k := 2^{-tq^2} \sum_{\boldsymbol{\sigma}} D_{\boldsymbol{\sigma}}^{\otimes k}.
\]
There are absolute constants $a,c,\varepsilon_0 > 0$ such that
\[
\chi^2(\overline{D}_k,U_{t,q}^{\otimes k}) \le \frac19
\]
whenever
\[
\gamma = a \frac{\varepsilon}{\sqrt{tq}},
\qquad
0 < \varepsilon \le \varepsilon_0,
\]
and
\[
k \le c \frac{q\sqrt{t}}{\varepsilon^2}.
\]
\end{proposition}

\begin{proof}
By the standard mixture formula,
\[
\chi^2(\overline{D}_k,U_{t,q}^{\otimes k}) + 1
=
\E_{\boldsymbol{\sigma},\boldsymbol{\tau}}
\Bigl( \E[L_{\boldsymbol{\sigma}}L_{\boldsymbol{\tau}}] \Bigr)^k .
\]
Choose $a$ and $\varepsilon_0$ so that $a\varepsilon_0\le c_*$, where $c_*$ is the constant in Lemma~\ref{lem:key-logmgf}; then $\gamma\le c_*$.
Applying Lemma~\ref{lem:pair-overlap} gives
\[
\chi^2 + 1
\le
\exp(Cktq\gamma^4)
\cdot
\E_{\boldsymbol{\sigma},\boldsymbol{\tau}}
\exp\bigl( 4k\beta_q\gamma^2 X(\boldsymbol{\sigma},\boldsymbol{\tau}) \bigr).
\]
For independent uniform sign collections, the random variables
\[
\sigma^{(r)}_{uv}\tau^{(r)}_{uv}
\qquad (1\le r\le t,\ 1\le u,v\le q)
\]
are i.i.d.\ uniform $\{\pm1\}$ random variables. Hence
\[
X(\boldsymbol{\sigma},\boldsymbol{\tau})
\stackrel{d}{=}
\sum_{i=1}^{tq^2} \xi_i,
\qquad
\xi_i \in \{\pm1\} \text{ i.i.d.}
\]
and therefore
\[
\E e^{\lambda X} = \cosh(\lambda)^{tq^2} \le \exp\Bigl(\frac{tq^2\lambda^2}{2}\Bigr)
\qquad\text{for every }\lambda\in\R .
\]
Substituting $\lambda = 4k\beta_q\gamma^2$ yields
\[
\chi^2 + 1
\le
\exp\bigl(Cktq\gamma^4 + Ck^2 \beta_q^2 \gamma^4 tq^2\bigr).
\]
By Lemma~\ref{lem:beta-bounds}, $\beta_q \le 1/q$. Since $\gamma^2 = a^2 \varepsilon^2 /(tq)$,
\[
ktq\gamma^4 = k \frac{a^4 \varepsilon^4}{tq},
\qquad
k^2 \beta_q^2 \gamma^4 tq^2 \le k^2 \frac{a^4 \varepsilon^4}{tq^2}.
\]
If $k \le c q\sqrt{t}/\varepsilon^2$, then
\[
ktq\gamma^4 \le c a^4 \varepsilon^2,
\qquad
k^2 \beta_q^2 \gamma^4 tq^2 \le c^2 a^4 .
\]
Choosing $a>0$ and $c>0$ sufficiently small and then $\varepsilon_0>0$ sufficiently small makes the exponent at most $\log(10/9)$, which implies $\chi^2 \le 1/9$.
\end{proof}

\begin{proof}[Proof of Proposition~\ref{prop:main-Gtq}]
Fix $t,q$ and a target distance parameter $\varepsilon \le c_{\mathrm{far}}\varepsilon_0$, where $c_{\mathrm{far}}$ is the constant from Proposition~\ref{prop:alt-far}. Apply Propositions~\ref{prop:alt-far} and \ref{prop:mixture-chi2} with the parameter
\[
\varepsilon' := \frac{\varepsilon}{c_{\mathrm{far}}}.
\]
Thus
\[
\gamma = a \frac{\varepsilon'}{\sqrt{tq}}
\]
and every distribution in the finite family
\[
\mathcal{A} := \{ D_{\boldsymbol{\sigma}} : \boldsymbol{\sigma} \in (\{\pm1\}^{q\times q})^t \}
\]
satisfies
\[
\TV(D_{\boldsymbol{\sigma}},U_{t,q}) \ge \varepsilon .
\]
By Proposition~\ref{prop:mixture-chi2}, if
\[
k \le c \frac{q\sqrt{t}}{(\varepsilon')^2}
= c\,c_{\mathrm{far}}^2 \frac{q\sqrt{t}}{\varepsilon^2},
\]
then the mixture over signs $\overline{D}_k$ of the $k$-sample distributions satisfies
\[
\TV(\overline{D}_k,U_{t,q}^{\otimes k})
\le
\sqrt{\frac12 \chi^2(\overline{D}_k,U_{t,q}^{\otimes k})}
\le
\frac14 .
\]
Therefore no tester can distinguish $U_{t,q}^{\otimes k}$ from this mixture over no-instances with average success probability greater than
\[
\frac12+\frac12\TV(\overline{D}_k,U_{t,q}^{\otimes k})
\le
\frac12+\frac{1}{12}
=
\frac{7}{12}.
\]
Equivalently, average success probability $2/3$ would require
$\TV(\overline{D}_k,U_{t,q}^{\otimes k})\ge 1/3$, whereas the bound above gives TV distance at
most $1/4$.  In particular, a tester that succeeded with probability at least $2/3$ for every
no-instance in the support of the mixture would have average success probability at least $2/3$
against the mixture, a contradiction. Since the mixture is supported on no-instances, the same
lower bound applies to the original composite testing problem. Renaming the constant
$c\,c_{\mathrm{far}}^2$ as $c_0$ completes the proof.
\end{proof}

\begin{proof}[Proof of Theorem~\ref{thm:main-lower}]
Fix integers $n,m$ with $n\le m\le n^2$.  Set
\[
q:=\max\left\{1,\left\lfloor \frac{m}{n}\right\rfloor\right\},
\qquad
t:=\max\left\{1,\left\lfloor \frac{n}{8q}\right\rfloor\right\},
\qquad
G_{n,m}:=G_{t,q}.
\]
We verify the size parameters.  If $q<n/16$, then $t=\Theta(n/q)$, and hence
\[
|V(G_{n,m})|=4tq=\Theta(n),
\qquad
|E(G_{n,m})|=4tq^2=\Theta(nq).
\]
Moreover, if $m<2n$ then $q=1$ and $nq=\Theta(m)$, while if $m\ge 2n$ then
$q=\Theta(m/n)$ and again $nq=\Theta(m)$.  Thus in this case
\[
|V(G_{n,m})|=\Theta(n),
\qquad
|E(G_{n,m})|=\Theta(m),
\qquad
q\sqrt t=\Theta(\sqrt{nq})=\Theta(\sqrt m).
\]

If $q\ge n/16$, then $q=\Theta(n)$ and therefore $m=\Theta(n^2)$.  Also $t=\Theta(1)$, so
\[
|V(G_{n,m})|=4tq=\Theta(n),
\qquad
|E(G_{n,m})|=4tq^2=\Theta(n^2)=\Theta(m),
\qquad
q\sqrt t=\Theta(n)=\Theta(\sqrt m).
\]

Applying Proposition~\ref{prop:main-Gtq} to $G_{n,m}=G_{t,q}$ and absorbing the absolute constants
in the comparison $q\sqrt t=\Theta(\sqrt m)$ gives the claimed lower bound.
\end{proof}

\paragraph{Acknowledgements.}
We thank Cl\'ement Canonne for helpful comments on an earlier version of this manuscript.

\bibliographystyle{plain}
\bibliography{main}

\begin{thebibliography}{10}

\bibitem{AcharyaDK15}
Jayadev Acharya, Constantinos Daskalakis, and Gautam Kamath.
\newblock Optimal testing for properties of distributions.
\newblock In {\em Advances in Neural Information Processing Systems ({NIPS})},
  pages 3591--3599, 2015.

\bibitem{AM18}
Hassan Ashtiani and Abbas Mehrabian.
\newblock Some techniques in density estimation, 2018.
\newblock arXiv:1801.04003.

\bibitem{BatuFRSW00}
Tugkan Batu, Lance Fortnow, Ronitt Rubinfeld, Warren~D. Smith, and Patrick
  White.
\newblock Testing that distributions are close.
\newblock In {\em Proceedings of the Annual IEEE Symposium on Foundations of
  Computer Science ({FOCS})}, pages 259--269, 2000.

\bibitem{BGKNT07}
Mohsen Bayati, David Gamarnik, Dimitriy~A. Katz, Chandra Nair, and Prasad
  Tetali.
\newblock Simple deterministic approximation algorithms for counting matchings.
\newblock In {\em Proceedings of the 39th Annual ACM Symposium on Theory of
  Computing ({STOC})}, pages 122--127, 2007.

\bibitem{BBCSV20}
Ivona Bez{\'a}kov{\'a}, Antonio Blanca, Zongchen Chen, Daniel {\v
  S}tefankovi{\v c}, and Eric Vigoda.
\newblock Lower bounds for testing graphical models: Colorings and
  antiferromagnetic {Ising} models.
\newblock {\em Journal of Machine Learning Research}, 21(25):1--62, 2020.

\bibitem{BhattacharyyaCY22}
Arnab Bhattacharyya, Cl{\'{e}}ment~L. Canonne, and Joy~Qiping Yang.
\newblock Independence testing for bounded degree bayesian networks.
\newblock In {\em Advances in Neural Information Processing Systems
  ({NeurIPS})}, 2022.

\bibitem{BhattacharyyaGGSV25}
Arnab Bhattacharyya, Sutanu Gayen, Philips George~John, Sayantan Sen, and N.~V.
  Vinodchandran.
\newblock Distribution learning meets graph structure sampling.
\newblock In {\em Advances in Neural Information Processing Systems
  ({NeurIPS})}, 2025.

\bibitem{BhattacharyyaGPTV23}
Arnab Bhattacharyya, Sutanu Gayen, Eric Price, Vincent Y.~F. Tan, and N.~V.
  Vinodchandran.
\newblock Near-optimal learning of tree-structured distributions by chow and
  liu.
\newblock {\em SIAM Journal on Computing}, 52(3):761--793, 2023.

\bibitem{BCSV21}
Antonio Blanca, Zongchen Chen, Daniel {\v S}tefankovi{\v c}, and Eric Vigoda.
\newblock Hardness of identity testing for restricted {Boltzmann} machines and
  {Potts} models.
\newblock {\em Journal of Machine Learning Research}, 22(152):1--56, 2021.

\bibitem{Canonne20}
Cl{\'{e}}ment~L. Canonne.
\newblock {\em A Survey on Distribution Testing: Your Data is Big. But is it
  Blue?}
\newblock Number~9 in Graduate Surveys. Theory of Computing Library, 2020.

\bibitem{Canonne22}
Cl{\'{e}}ment~L. Canonne.
\newblock Topics and techniques in distribution testing: {A} biased but
  representative sample.
\newblock {\em Foundations and Trends in Communications and Information
  Theory}, 19(6):1032--1198, 2022.

\bibitem{CDKS20}
Cl{\'e}ment~L. Canonne, Ilias Diakonikolas, Daniel~M. Kane, and Alistair
  Stewart.
\newblock Testing {Bayesian} networks.
\newblock {\em IEEE Transactions on Information Theory}, 66(5):3132--3170,
  2020.

\bibitem{canonne2026entropy}
Cl{\'e}ment~L Canonne, Yash Pote, Jonathan Scarlett, and Joy~Qiping Yang.
\newblock Entropy equivalence testing.
\newblock {\em arXiv preprint arXiv:2605.23225}, 2026.

\bibitem{CY24}
Cl{\'e}ment~L Canonne and Joy~Q Yang.
\newblock Entropy testing and its application to testing {B}ayesian networks.
\newblock {\em Advances in Neural Information Processing Systems}, 2024.

\bibitem{CG23}
Zongchen Chen and Yuzhou Gu.
\newblock Fast sampling of {$b$}-matchings and {$b$}-edge covers.
\newblock In {\em Proceedings of the 35th Annual ACM-SIAM Symposium on Discrete
  Algorithms ({SODA})}, pages 4972--4987, 2024.

\bibitem{CLV23}
Zongchen Chen, Kuikui Liu, and Eric Vigoda.
\newblock Optimal mixing of {Glauber} dynamics: Entropy factorization via
  high-dimensional expansion.
\newblock {\em SIAM Journal on Computing}, pages STOC21--104--STOC21--153,
  October 2023.
\newblock Open access at \url{https://epubs.siam.org/doi/10.1137/21M1443340}.

\bibitem{ChengDKS18}
Yu~Cheng, Ilias Diakonikolas, Daniel~M. Kane, and Alistair Stewart.
\newblock Robust learning of fixed-structure {Bayesian} networks.
\newblock In {\em Advances in Neural Information Processing Systems 31}, pages
  10304--10316, 2018.

\bibitem{DDK18}
Constantinos Daskalakis, Nishanth Dikkala, and Gautam Kamath.
\newblock Testing {Ising} models.
\newblock In {\em Proceedings of the 29th Annual ACM-SIAM Symposium on Discrete
  Algorithms ({SODA})}, pages 1989--2007, 2018.

\bibitem{DaskalakisP17}
Constantinos Daskalakis and Qinxuan Pan.
\newblock Square hellinger subadditivity for bayesian networks and its
  applications to identity testing.
\newblock In {\em {COLT}}, volume~65 of {\em Proceedings of Machine Learning
  Research}, pages 697--703. {PMLR}, 2017.

\bibitem{DL01}
Luc Devroye and G{\'a}bor Lugosi.
\newblock {\em Combinatorial Methods in Density Estimation}.
\newblock Springer, 2001.

\bibitem{DMR20}
Luc Devroye, Abbas Mehrabian, and Tommy Reddad.
\newblock The minimax learning rates of normal and {Ising} undirected graphical
  models.
\newblock {\em Electronic Journal of Statistics}, 14(1):2338--2361, 2020.

\bibitem{HL72}
Ole~J. Heilmann and Elliott~H. Lieb.
\newblock Theory of monomer-dimer systems.
\newblock {\em Communications in Mathematical Physics}, 25:190--232, 1972.

\bibitem{JS89}
Mark Jerrum and Alistair Sinclair.
\newblock Approximating the permanent.
\newblock {\em SIAM Journal on Computing}, 18(6):1149--1178, 1989.

\bibitem{jerrum1996markov}
Mark Jerrum and Alistair Sinclair.
\newblock The markov chain monte carlo method: an approach to approximate
  counting and integration.
\newblock {\em Approximation algorithms for NP-hard problems}, pages 482--520,
  1996.

\bibitem{LelargeZhou14}
Marc Lelarge and Hang Zhou.
\newblock Sublinear-time algorithms for monomer-dimer systems on bounded degree
  graphs.
\newblock {\em Theoretical Computer Science}, 548:68--78, 2014.

\bibitem{Paninski08}
Liam Paninski.
\newblock A coincidence-based test for uniformity given very sparsely sampled
  discrete data.
\newblock {\em IEEE Transactions on Information Theory}, 54(10):4750--4755,
  2008.

\bibitem{Rubinfeld12}
Ronitt Rubinfeld.
\newblock Taming big probability distributions.
\newblock {\em XRDS: Crossroads, The ACM Magazine for Students}, 19(1):24--28,
  2012.

\bibitem{Yat85}
Yannis~G. Yatracos.
\newblock Rates of convergence of minimum distance estimators and
  {Kolmogorov}'s entropy.
\newblock {\em The Annals of Statistics}, 13(2):768--774, 1985.

\end{thebibliography}

\appendix
\section{Learning the Matching Gibbs Family}
\label{sec:learning}

We next record the information-theoretic sample complexity of learning the
entire Gibbs distribution over matchings.  This should be contrasted with the
uniformity-testing guarantee above.  The learning bound below is a direct
application of the classical \emph{Yatracos minimum-distance method}; see
Yatracos~\cite{Yat85}, Devroye and Lugosi~\cite{DL01}, and the exposition
of Ashtiani and Mehrabian~\cite[Theorem~10 and Corollary~11]{AM18}.
Closely related VC/Yatracos arguments are also used for learning Ising and
normal graphical models in~\cite{DMR20}.

For a fixed graph $G=(V,E)$, write
\[
\mathcal F_G := \{D_w:\ w\in \R_{>0}^E\}.
\]
A learner receives i.i.d. samples from an unknown member of $\mathcal F_G$
and must output a distribution close to it in total variation distance.
The result below is proper, in the sense that the output distribution is
again a member of $\mathcal F_G$.  It is, however, only an information-theoretic
statement; no computational efficiency of the minimum-distance estimator is
claimed.

\begin{proposition}[Information-theoretic proper learning upper bound]
\label{prop:learning-upper-bound}
There is a universal constant $C<\infty$ such that the following holds.
For every graph $G=(V,E)$ with $m=|E|$, every
$\varepsilon,\delta\in(0,1)$, and every unknown $w\in\R_{>0}^E$, there exists
a proper learner which, given
\[
n \ge C\,\frac{m+\log(1/\delta)}{\varepsilon^2}
\]
i.i.d. samples from $D_w$, outputs a distribution $D_{\widehat w}\in\mathcal F_G$
satisfying
\[
\TV(D_{\widehat w},D_w)\le \varepsilon
\]
with probability at least $1-\delta$.
In particular, the learning upper bound is independent of the maximum degree
of $G$ and of the magnitude of the weights.
\end{proposition}

\begin{proof}
Let $\Omega:=\Mcal(G)$ be the domain.  For a class
$\mathcal A\subseteq 2^\Omega$ and two probability distributions $P,Q$ on
$\Omega$, write
\[
\|P-Q\|_{\mathcal A}:=\sup_{A\in\mathcal A}|P(A)-Q(A)|.
\]
The \emph{Yatracos class} of $\mathcal F_G$ is
\[
\mathcal Y_G
:=
\bigl\{
\{M\in\Mcal(G): D_w(M)\ge D_v(M)\}:
w,v\in\R_{>0}^E
\bigr\}.
\]

We first bound the VC dimension of $\mathcal Y_G$.  Identify a matching
$M$ with its edge-incidence vector
\[
x(M):=(\one\{e\in M\})_{e\in E}\in\{0,1\}^m.
\]
For any two weight vectors $w,v\in\R_{>0}^E$,
\[
D_w(M)\ge D_v(M)
\]
is equivalent to
\[
\sum_{e\in E}(\log w_e-\log v_e)\one\{e\in M\}
\ge
\log Z(w)-\log Z(v).
\]
Thus every set in $\mathcal Y_G$ is the restriction to
$\{x(M):M\in\Mcal(G)\}\subseteq\{0,1\}^m$ of an affine halfspace in $\R^m$.
Since affine halfspaces in $\R^m$ have VC dimension $m+1$, restriction to a
subset cannot increase VC dimension, and therefore
\[
\mathrm{VC}(\mathcal Y_G)\le m+1.
\]

Let $P:=D_w$ be the unknown target distribution and let $P_n$ be the empirical
distribution of the $n$ samples.  By VC uniform convergence, with probability
at least $1-\delta$,
\[
\|P_n-P\|_{\mathcal Y_G}
\le
C_0\sqrt{\frac{m+\log(1/\delta)}{n}}
\]
for a universal constant $C_0$.

Let $\widehat P\in\mathcal F_G$ be an $\eta$-approximate empirical Yatracos
minimizer:
\[
\|\widehat P-P_n\|_{\mathcal Y_G}
\le
\inf_{Q\in\mathcal F_G}\|Q-P_n\|_{\mathcal Y_G}+\eta .
\]
Such an approximate minimizer exists by the definition of the infimum.  Since
the target $P$ itself belongs to $\mathcal F_G$,
\[
\|\widehat P-P_n\|_{\mathcal Y_G}
\le
\|P-P_n\|_{\mathcal Y_G}+\eta .
\]
Moreover, because both $P$ and $\widehat P$ belong to $\mathcal F_G$, the
Yatracos class contains the set on which $P-\widehat P$ is nonnegative, and
therefore
\[
\TV(P,\widehat P)=\|P-\widehat P\|_{\mathcal Y_G}.
\]
Hence, on the uniform-convergence event,
\[
\begin{aligned}
\TV(P,\widehat P)
&=
\|P-\widehat P\|_{\mathcal Y_G} \\
&\le
\|P-P_n\|_{\mathcal Y_G}
+
\|P_n-\widehat P\|_{\mathcal Y_G} \\
&\le
2\|P-P_n\|_{\mathcal Y_G}+\eta \\
&\le
2C_0\sqrt{\frac{m+\log(1/\delta)}{n}}+\eta .
\end{aligned}
\]
Taking $\eta\le \varepsilon/3$ and
$n\ge C(m+\log(1/\delta))/\varepsilon^2$ for a sufficiently large universal
constant $C$ gives the claimed bound.  Writing
$\widehat P=D_{\widehat w}$ completes the proof.
\end{proof}

\begin{remark}[Tightness of the learning sample complexity]
\label{rem:learning-lower-bound}
The linear dependence on $m$ in Proposition~\ref{prop:learning-upper-bound}
is unavoidable, even for graphs of maximum degree $1$.  Indeed, let
\[
G_m:=\bigsqcup_{i=1}^m K_2
\]
be the disjoint union of $m$ isolated edges.  Then every matching is identified
with a vector in $\{0,1\}^m$, and the Gibbs distribution over matchings factorizes as
\[
D_w
=
\bigotimes_{i=1}^m
\Ber\!\left(\frac{w_i}{1+w_i}\right).
\]
Thus the Gibbs family over matchings realizes all $m$-dimensional Bernoulli product
distributions with success probabilities in $(0,1)$.  The standard minimax
lower bound for learning Bernoulli product distributions in total variation
distance already holds in, say, a fixed neighborhood of the all-$1/2$ parameter
vector, and is
\[
\Omega\!\left(\frac{m}{\varepsilon^2}\right)
\]
samples for constant success probability; equivalently, this is the in-degree-zero case of the
lower bound for learning fixed-structure Bernoulli Bayesian networks stated in~\cite{ChengDKS18}.
Consequently, the minimax learning
sample complexity of the Gibbs family over matchings is
\[
\Theta\!\left(\frac{m}{\varepsilon^2}\right)
\]
for constant failure probability, while the uniformity-testing sample complexity
proved above is of order $\sqrt m/\varepsilon^2$ up to the stated
degree-dependent and logarithmic factors.
\end{remark}

\end{document}